\documentclass[a4paper,11pt]{article}
\usepackage{bm,mathrsfs,amsmath,amsthm,amssymb,amscd,longtable,array,graphicx}
\newcommand{\nc}{\newcommand}
\nc{\rnc}{\renewcommand}
\nc{\nn}{\nonumber}
\nc{\der}{{\partial}}
\rnc{\Im}{{\rm{Im}\,}}
\rnc{\Re}{{\rm{Re}\,}}
\nc{\db}{\displaybreak[0]\\}
\nc{\bra}{\langle}
\nc{\ket}{\rangle}
\nc{\bs}{\boldsymbol}

\newtheorem{theorem}{Theorem}[section]
\newtheorem{lemma}[theorem]{Lemma}

\newtheorem{proposition}[theorem]{Proposition}

\theoremstyle{definition}
\newtheorem{definition}[theorem]{Definition}

\numberwithin{equation}{section}

\numberwithin{equation}{section}

\begin{document}%
%
\title{Elliptic supersymmetric integrable model \\
and multivariable elliptic functions}

\author{
Kohei Motegi \thanks{E-mail: kmoteg0@kaiyodai.ac.jp}
\\\\
{\it Faculty of Marine Technology, Tokyo University of Marine Science and Technology,}\\
 {\it Etchujima 2-1-6, Koto-Ku, Tokyo, 135-8533, Japan} \\
\\\\
\\
}

\date{\today}

\maketitle

\begin{abstract}
We investigate the elliptic integrable model
introduced by Deguchi and Martin, which is an elliptic extension
of the Perk-Schultz model.
We introduce and study a class of partition functions of the elliptic model
by using the Izergin-Korepin analysis.
We show that the partition functions are expressed as a product of
elliptic factors and elliptic Schur-type symmetric functions.
This result resembles the recent works by number theorists
in which the correspondence between the partition functions
of trigonometric models and
the product of the deformed Vandermonde determinant and
Schur functions were established.
\end{abstract}

\section{Introduction}
In statistical physics and field theory,
integrable models \cite{Bethe,FST,Baxter,KBI} play special roles
not only because many physical quantities can be computed exactly
but also because of its deep connections with mathematics.
One of the highlights is the discovery of quantum groups \cite{Dr,J}
through the investigations of the algebraic structure
of the $R$-matrix.
From the point of view of integrable models,
quantum groups are related with trigonometric models
whose matrix elements are given in terms of trigonometric functions.
There is also a class of elliptic integrable models whose matrix elements
are given in terms of elliptic functions.
The most famous one is the Baxter's eight-vertex model \cite{eight}.
Investigating the underlying algebraic structures,
several versions of
the elliptic quantum groups have been formulated and studied
\cite{Felder,FV,Konno1,Konno2,JKOS}.

In statistical physics, the most important objects are
partition functions.
Partition functions of integrable models are global objects
constructed from $R$-matrices.
By developing various methods such as the quantum inverse scattering method,
there are now extensive studies on the
partition functions of the elliptic eight-vertex solid-on-solid model
\cite{Ros,PRS,FK,YZ,Chinesegroup,Chinesegroup2,Galleasone,Galleastwo,Borodin},
ranging from the domain wall boundary partition functions to the wavefunctions.

The elliptic model we mentioned above is the eight-vertex solid-on-solid model
which is related with the Baxter's eight-vertex model by the vertex-face transform.
In this paper, we investigate another type of
elliptic integrable model, which was introduced by
Deguchi and Martin \cite{DM}.
Deguchi and Martin introduced this elliptic model
as an elliptic extension of the Perk-Schultz model \cite{PS}.
The Perk-Schultz model is a trigonometric model,
which was later found to have
the quantum superalgebra structure \cite{Yamane}.
Hence we sometimes call the Perk-Schultz model
as supersymmetric integrable model.

The Deguchi-Martin model was constructed as an elliptic generalization
of the face-type version of the Perk-Schultz model.
The quantum supergroup structure seems not to be explored till now.
For the case of the ordinary elliptic model,
the elliptic quantum group structure is understood to be obtained by
twisting the quantum group associated with the trigonometric models
\cite{JKOS}.
It might be possible to understand the underlying algebraic structure
of the Deguchi-Martin model as an elliptic quantum supergroup
obtained from a quasi-Hopf twist of the quantum supergroup
of the trigonometric Perk-Schultz model.

The main motivation for studying elliptic versions of supersymmetric models
in this paper comes from the previous works on the
trigonometric version of the model. In recent years,
trigonometric models were found to have
connections with the field of algebraic combinatorics.
Number theorists Bump-Brubaker-Friedberg \cite{BBF} found that
a class of partition functions of a free-fermion model in an external field
gives rise to an integrable model realization of the
Tokuyama formula for the Schur functions \cite{To}
(see also \cite{OkTo,HK1,HK2} for pioneering works on variations
of the Tokuyama formula).
Tokuyama formula is a one-parameter deformation of the Weyl character formula
for the Schur functions.
This fundamental result lead people to
find generalizations and variations of the
Tokuyama-type formula for various types of symmetric functions
\cite{Iv,BBCG,Tabony,BMN,HK,BS,BBB,LMP,dualsymplectic}.
The latest topic is the introduction of the notion of
the metaplectic ice, which is explicitly constructed in \cite{BBB}
by twisting the higher rank Perk-Schultz model.

As for the domain wall boundary partition function,
the celebrated Izergin-Korepin analysis \cite{Ko,Iz}
was performed, and factorization phenomena on
the domain wall boundary partition functions
was found for the Perk-Schultz model and the
closely related Felderhof free-fermion model
\cite{Felderhof} in \cite{ZZ,FCWZ}
(see also \cite{ZYZ} for an application to correlation functions).
This factorization formula for the domain wall boundary partition functions
cannot be observed in general
for the case of the $U_q(sl_2)$ six-vertex models.
Instead, a determinant formula called as the Izergin-Korepin determinant
was already obtained in 1980s \cite{Ko,Iz}.
The factorization phenomena for the Perk-Schultz and Felderhof models
was extended to the elliptic models \cite{FWZ,Zuparic}.
For example, an extensive study on the domain wall boundary partition
functions of the Deguchi-Martin model was performed
in the PhD thesis of Zuparic \cite{Zuparic}.

In this paper, focusing on the Deguchi-Martin model,
we introduce and investigate the explicit forms of
a more general class of partition functions which
includes the domain wall boundary partition functions as a special case.
We apply the Izergin-Korepin technique to study the partition functions
which was recently applied to the $U_q(sl_2)$ six-vertex model in \cite{Motr}.

By using the Izergin-Korepin technique,
We prove that the partition functions are expressed
as a certain multivariable elliptic functions.
Let us briefly present this main result of this paper in some more detail.
We denote the partition functions of the Deguchi-Martin model treated in this
paper as
$W_{M,N}(u_1,\dots,u_N|v_1,\dots,v_M|x_1,\dots,x_N|a_{12})$.
The precise definition will be given in Definition \ref{definitionofwavefunctions}
in section 2. The partition functions are defined using the local elliptic
weights of the Deguchi-Martin model,
and are determined by the boundary conditions.
Since the local weights of the Deguchi-Martin model
will be described using spectral parameters and the state vectors,
the partition functions depend on two classes of complex parameters
$u_1,\dots,u_N$, $v_1,\dots,v_M$ coming from the spectral parameters,
and a complex parameter $a_{12}$ coming from the state vectors.
The partition functions we investigate in this paper also
have dependence on boundary conditions, which are labelled using a sequence of
integers $x_1,\dots,x_N$ satisfying
$1 \le x_1 < \cdots < x_N \le M$.
The dependence on these parameters are why we denote the partition functions as
$W_{M,N}(u_1,\dots,u_N|v_1,\dots,v_M|x_1,\dots,x_N|a_{12})$.

The main result of this paper is Theorem \ref{maintheoremstatement}
in section 4, which states that
the partition functions is expressed as a product of
elliptic factors and an elliptic version of the Schur functions as
\begin{align}
&W_{M,N}(u_1,\dots,u_N|v_1,\dots,v_M|x_1,\dots,x_N|a_{12})
\nonumber \\
=&\prod_{1 \le j < k \le N} \frac{[1+u_k-u_j]}{[1]}
\prod_{k=1}^N \prod_{x_{k+1}^M-2 \ge j \ge x_k^M}
\frac{[a_{12}+j+N]}{[a_{12}+j+N-k]} \nonumber \\
&\times \sum_{\sigma \in S_N}
\prod_{1 \le j < k \le N} \frac{[1]}{[u_{\sigma(j)}-u_{\sigma(k)}]}
\prod_{j=1}^N \prod_{k=x_j+1}^M \frac{[u_{\sigma(j)}-v_k]}{[1]}
\nonumber \\
&\times
\prod_{j=1}^N \frac{[-u_{\sigma(j)}+v_{x_j}+a_{12}+x_j+N-2]}{[a_{12}+x_j+N-2]}
\prod_{j=1}^N \prod_{k=1}^{x_j-1}
\frac{[1+u_{\sigma(j)}-v_k]}{[1]}.
\end{align}
Here, $[u]=H(\lambda u)$ where
\begin{align}
H(u)=2q^{\frac{1}{4}} \mathrm{sin} \Bigg( \frac{\pi u}{2 K_1} \Bigg)
\prod_{n=1}^\infty \Bigg\{ 1-2q^{2n} \mathrm{cos} \Bigg( \frac{\pi u}{K_1} \Bigg)+q^{4n} \Bigg\} \{1-q^{2n} \},
\end{align}
is an elliptic theta function, and
$x_k^M$, $k=1,\dots,N+1$ is defined as $x_{N+1}^M=M+1$
and $x_k^M=x_k$, $k=1,\dots,N$.

The above result on the correspondence between the partition functions
and elliptic multivariable functions
resembles the one for the trigonometric model,
whose partition functions
were found to be given as the product of a one-parameter deformation
of the Vandermonde determinant and the (factorial)
Schur functions \cite{BBF,BMN},
hence can be viewed as an elliptic analogue of the result
for the trigonometric model.

This paper is organized as follows.
In the next section, we first review the properties
of elliptic functions, and introduce the Deguchi-Martin model
and a class of partition functions of the model.
We make the Izergin-Korepin analysis in section 3
and list the properties needed to uniquely determine the explicit form
of the partition functions.
In section 4, we present the main theorem of this paper,
the explicit expression of the partition functions
as a certain product of
elliptic factors and elliptic symmetric functions.
In section 5, we prove the main theorem by showing that the elliptic functions
satisfies all the required properties in the
Izergin-Korepin analysis which uniquely defines the partition functions.
Section 6 is devoted to the conclusion of this paper.

\section{Elliptic functions and Deguchi-Martin model}
In this section,
we first review the properties of elliptic functions.
Next, we introduce the Deguchi-Martin model,
following Deguchi-Martin \cite{DM}
and the PhD thesis of Zuparic \cite{Zuparic}.

\subsection{Elliptic functions}
We first introduce elliptic functions and list the properties
which we use in this paper.
The half period magnitudes $K_1$ and $K_2$ are defined for 
elliptic nome $q$ $(0 < q < 1)$ as
\begin{align}
K_1&=\frac{1}{2}\pi \prod_{n=1}^\infty
\Bigg( \Bigg\{ \frac{1+q^{2n-1}}{1-q^{2n-1}} \Bigg\}
\Bigg\{ \frac{1-q^{2n}}{1+q^{2n}} \Bigg\} \Bigg)^2, \\
K_2&=-\frac{1}{\pi}K_1 \mathrm{log}(q).
\end{align}
The theta functions $H(u)$ is defined using $K_1$, $K_2$ and $q$ as
\begin{align}
H(u)=2q^{\frac{1}{4}} \mathrm{sin} \Bigg( \frac{\pi u}{2 K_1} \Bigg)
\prod_{n=1}^\infty \Bigg\{ 1-2q^{2n} \mathrm{cos} \Bigg( \frac{\pi u}{K_1} \Bigg)+q^{4n} \Bigg\} \{1-q^{2n} \}.
\end{align}
The important properties of the theta functions are the following quasi-periodicities
\begin{align}
H(u+2mK_1)&=(-1)^m H(u), \\
H(u+2inK_2)&=(-1)^n q^{-n^2}
\mathrm{exp} \Bigg(-\frac{in\pi u}{K_1} \Bigg) H(u),
\end{align}
for integers $m$ and $n$.
Note also that $H(u)$ is an odd function $H(-u)=-H(u)$.

In their work of the analysis on the domain wall boundary partition functions
of elliptic integrable models,
Pakuliak-Rubtsov-Silantyev \cite{PRS} (see also Felder-Schorr \cite{FS})
used the following notions and properties of elliptic polynomials.

A character is a group homomorphism
$\chi$ from multiplicative groups
$\Gamma=\mathbf{Z}+\tau \mathbf{Z}$ to $\mathbf{C}^\times$.
An $N$-dimensional space $\Theta_N(\chi)$
is defined for each character $\chi$ and positive integer $N$,
which consists of holomorphic functions $\phi(y)$ on $\mathbf{C}$
satisfying the quasi-periodicities
\begin{align}
\phi(y+1)&=\chi(1) \phi(y), \label{propertyuseone} \\
\phi(y+\tau)&=\chi(\tau) e^{-2 \pi i Ny-\pi i N \tau}\phi(y).
\label{propertyusetwo}
\end{align}
The elements of the space $\Theta_N(\chi)$ are called elliptic polynomials.
The space $\Theta_N(\chi)$ is $N$-dimensional,
and we use the following property for the elliptic polynomials.
\begin{proposition} \cite{PRS,FS} \label{propositionelliptic}
Suppose there are two elliptic polynomials $P(y)$ and $Q(y)$
in $\Theta_N(\chi)$, where $\chi(1)=(-1)^N$, $\chi(\tau)=(-1)^N e^\alpha$.
If those two polynomials are equal $P(y_j)=Q(y_j)$
at $N$ points $y_j$, $j=1,\dots,N$ satisfying
$y_j-y_k \not\in \Gamma$, $\sum_{k=1}^N y_k-\alpha \not\in \Gamma$,
then the two polynomials are exactly the same $P(y)=Q(y)$.
\end{proposition}
The above proposition is an elliptic version of the fact that
if $P(y)$ and $Q(y)$ are polynomials of degree $N-1$ in $y$,
and if these polynomials match at $N$ distinct points,
then the two polynomials are exactly the same.
This property was used for the analysis on the
domain wall boundary partition functions of the
eight-vertex solid-on-solid model \cite{ABF},
and we will also use this for the Izergin-Korepin anlaysis on
the partition functions of the elliptic Deguchi-Martin model.

\subsection{Deguchi-Martin model and Partition functions}
We introduce the Deguchi-Martin (elliptic supersymmetric integrable)
model.
Deguchi-Martin model is a class of face model which is an
elliptic generalization of the trigonometric Perk-Schultz vertex model.
The vertex and face models both have the
spectral parameters and inhomogeneous parameters.
The face models have additional state variables
coming from the state vectors.
The state variables are functions of the state vectors of the face model.
For the case of the Deguchi-Martin model treated in this paper,
the state vectors are elements of $\mathbf{Z}^2$.
Let us introduce two unit vectors
$\hat{e}_1=\{1,0 \}$ and $\hat{e}_2=\{0,1\}$ which form a basis
of $\mathbf{Z}^2$.
We also introduce an arbitrary complex variable $\omega_{12}$,
and define $\omega_{21}$ as $\omega_{12}=-\omega_{21}$.
We also set $\omega_{11},\omega_{22}$,$\epsilon_1$,$\epsilon_2$
as
$\omega_{11}=\omega_{22}=0$, $\epsilon_1=1$, $\epsilon_2=-1$
which we introduce for the purpose of defining the weights of the model.

Each face is labelled by the four state vectors
associated with the four vertices around the face.
To each face where the surrounding four state vectors are fixed,
we assign a weight which depends on
the spectral parameters and inhomogeneous parameters.
One can think that the spectral and inhomogeneous parameters
are carried by lines penetrating the vertical and horizontal edges.
We denote the weights associated with the state vectors
$\overrightarrow{a}, \overrightarrow{b},
\overrightarrow{c}, \overrightarrow{d}$,
spectral parameter $u$ and inhomogeneous parameter $v$
by (Figure \ref{pictureface})
\begin{align}
W \Bigg(
 \begin{array}{cc}
\overrightarrow{a} & \overrightarrow{b} \\
\overrightarrow{c} & \overrightarrow{d}
 \end{array} 
\Bigg|
u
\Bigg|
v
\Bigg). \label{notationweight}
\end{align}
To define the weights of the Deguchi-Martin model, one introduces the following notation
for convenience
\begin{align}
[u]=H(\lambda u),
\end{align}
for a fixed complex variable $\lambda$.
In this notation, the quasi-periodicities are expressed as
\begin{align}
\Bigg[ u+\frac{2K_1}{\lambda} \Bigg]&=-[u], \label{qpuseone} \\
\Bigg[ u+2i\frac{K_2}{\lambda} \Bigg]
&=-q^{-1} \mathrm{exp}
\Bigg(
-\frac{i \pi \lambda u}{K_1}
\Bigg)
[u]. \label{qpusetwo}
\end{align}

\begin{figure}[ht]
\includegraphics[width=12cm]{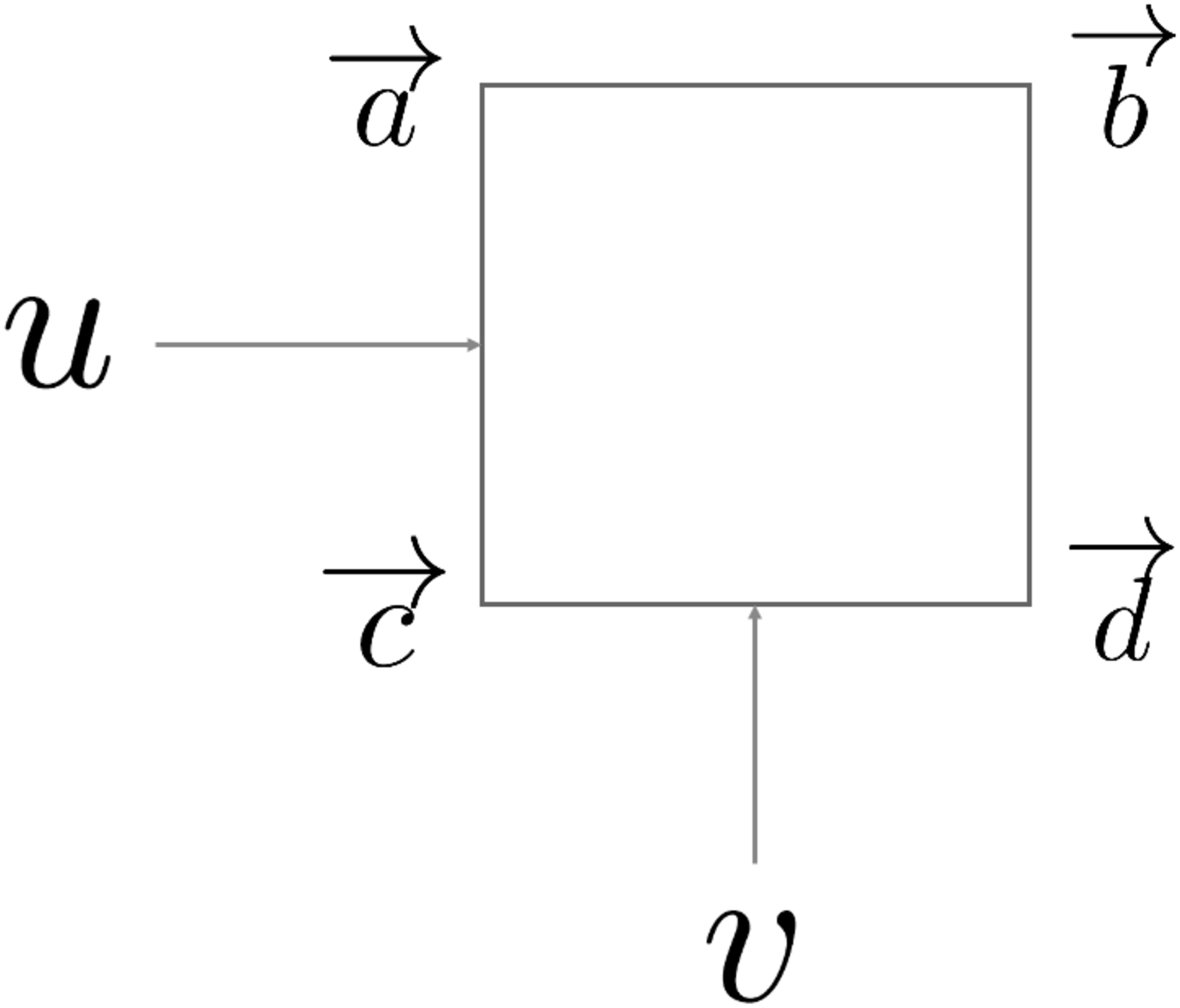}
\caption{A graphical description of a face of the Deguchi-Martin model.
The state vectors
are associated at four vertices,
and a spectral parameter $u$
and an inhomogenous parameter $v$ can be associated
with integrable lattice models.
To each face where the surrounding four state vectors
are fixed, we associate a weight
which we denote by \eqref{notationweight}.
}
\label{pictureface}
\end{figure}

The weights for the Deguchi-Martin model is defined as
(Figure \ref{pictureloperator})
\begin{align}
W \Bigg(
 \begin{array}{cc}
\overrightarrow{a} & \overrightarrow{a}+\hat{e}_j \\
\overrightarrow{a}+\hat{e}_j & \overrightarrow{a}+2\hat{e}_j
 \end{array} 
\Bigg|
u
\Bigg|
v
\Bigg)
=&\frac{[1+\epsilon_j(u-v)]}{[1]}, \ \ \ j=1,2, \label{weightsone} \\
W \Bigg(
 \begin{array}{cc}
\overrightarrow{a} & \overrightarrow{a}+\hat{e}_k \\
\overrightarrow{a}+\hat{e}_j & \overrightarrow{a}+\hat{e}_j+\hat{e}_k
 \end{array} 
\Bigg|
u
\Bigg|
v
\Bigg)
=&\frac{[u-v][a_{jk}-1]}{[1][a_{jk}]}, \ \ \ (j,k)=(1,2), (2,1), 
\label{weightstwo}
\\
W \Bigg(
 \begin{array}{cc}
\overrightarrow{a} & \overrightarrow{a}+\hat{e}_j \\
\overrightarrow{a}+\hat{e}_j & \overrightarrow{a}+\hat{e}_j+\hat{e}_k
 \end{array} 
\Bigg|
u
\Bigg|
v
\Bigg)
=&\frac{[a_{jk}-(u-v)]}{[a_{jk}]}, \ \ \ (j,k)=(1,2), (2,1),
\label{weightsthree}
\end{align}
when the state vector on the top left corner vertex
is $\overrightarrow{a} \in \mathbf{Z}^2$.
Here, $a_{jk}$ is a function of $\overrightarrow{a}=(a_1,a_2)$
defined as
$a_{jk}=\epsilon_j a_j-\epsilon_k a_k+\omega_{jk}$.
Since we are dealing with the $m=n=1$ case of the
$gl(m|n)$ Deguchi-Martin model,
the $a_{jk}$s appearing in this paper are $a_{12}=a_1+a_2+\omega_{12}$
and $a_{21}=-a_1-a_2+\omega_{21}$.
Moreover, since $\omega_{12}=-\omega_{21}$, we have the constraint
$a_{12}=-a_{21}$.

\begin{figure}[ht]
\includegraphics[width=15cm]{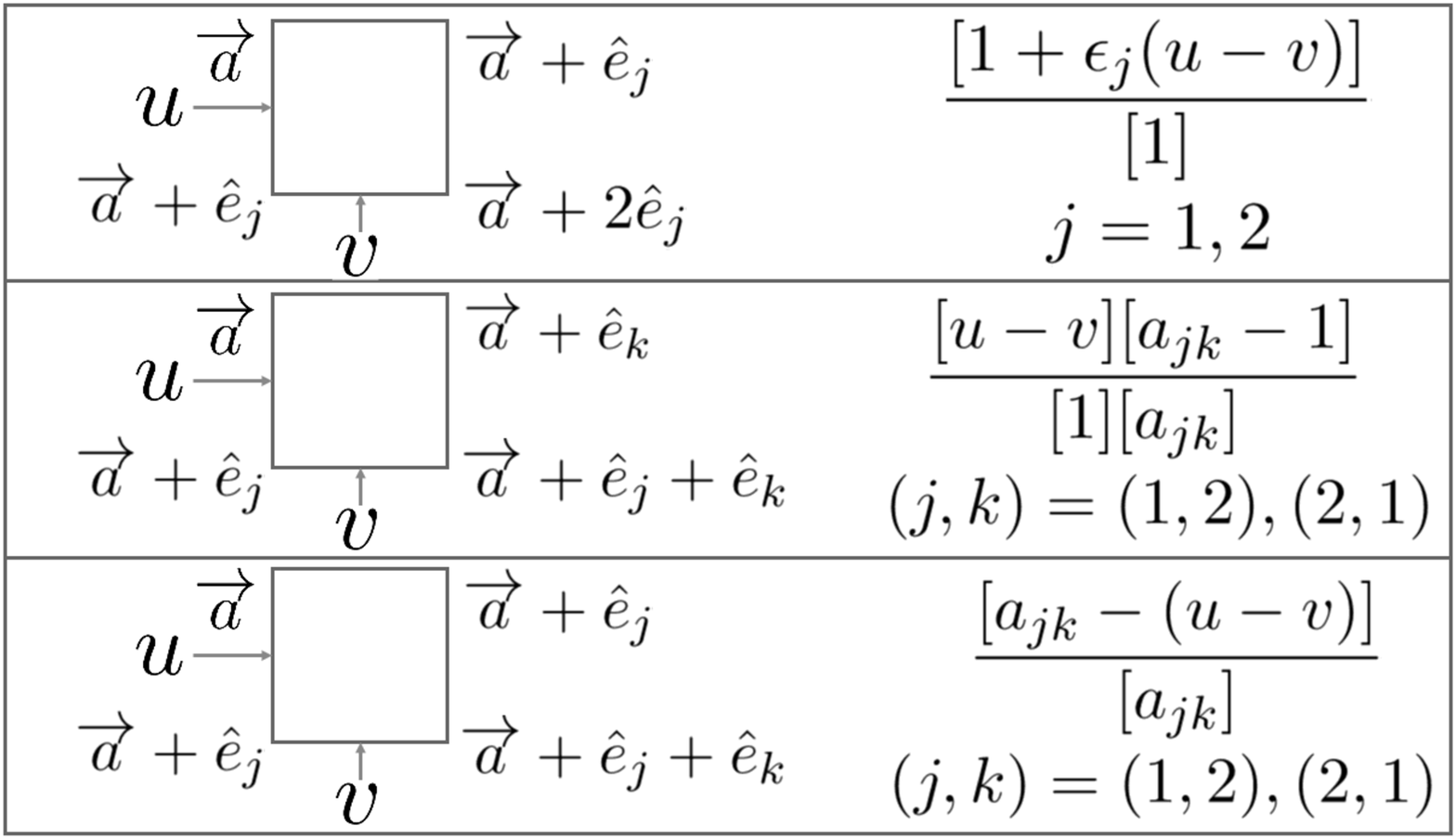}
\caption{The weights for the Deguchi-Martin model
\eqref{weightsone},
\eqref{weightstwo},
\eqref{weightsthree}
associated with three types of the configurations
of the state vectors.
}
\label{pictureloperator}
\end{figure}

\begin{figure}[ht]
\includegraphics[width=15cm]{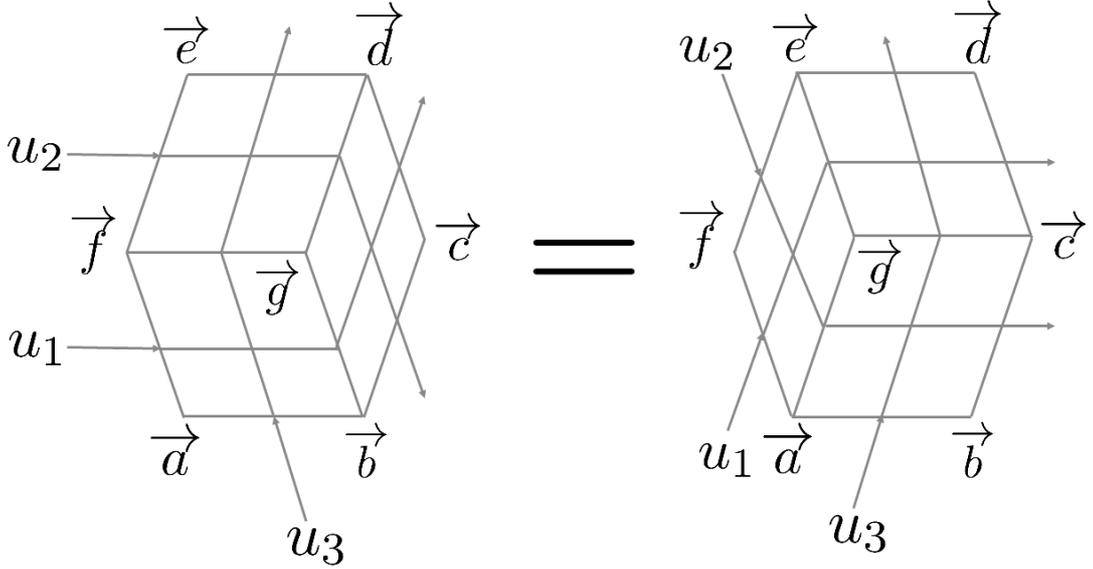}
\caption{The face-type Yang-Baxter relation (star-triangle relation) \eqref{yangbaxter}.
We take the sum over the inner state vectors.
}
\label{pictureyangbaxter}
\end{figure}

All the other weights
$\displaystyle W \Bigg(
 \begin{array}{cc}
\overrightarrow{a} & \overrightarrow{b} \\
\overrightarrow{c} & \overrightarrow{d}
 \end{array} 
\Bigg|
u
\Bigg|
v
\Bigg)$
where the tuple of the state vectors
($\overrightarrow{a}, \overrightarrow{b},
\overrightarrow{c}, \overrightarrow{d}$)
cannot be written in the form of
\eqref{weightsone},
\eqref{weightstwo},
\eqref{weightsthree} are defined to be zero.

The weights
\eqref{weightsone},
\eqref{weightstwo},
\eqref{weightsthree}
satisfy the face-type Yang-Baxter relation
(star-triangle relation) (Figure \ref{pictureyangbaxter})

\begin{align}
&
\sum_{\overrightarrow{g} \in \mathbf{Z}^2}
W \Bigg(
 \begin{array}{cc}
\overrightarrow{f} & \overrightarrow{g} \\
\overrightarrow{a} & \overrightarrow{b}
 \end{array} 
\Bigg|
u_1
\Bigg|
u_3
\Bigg)
W \Bigg(
 \begin{array}{cc}
\overrightarrow{e} & \overrightarrow{d} \\
\overrightarrow{f} & \overrightarrow{g}
 \end{array} 
\Bigg|
u_2
\Bigg|
u_3
\Bigg)
W \Bigg(
 \begin{array}{cc}
\overrightarrow{d} & \overrightarrow{c} \\
\overrightarrow{g} & \overrightarrow{b}
 \end{array} 
\Bigg|
u_2
\Bigg|
u_1
\Bigg)
\nonumber \\
=&
\sum_{\overrightarrow{g} \in \mathbf{Z}^2}
W \Bigg(
 \begin{array}{cc}
\overrightarrow{e} & \overrightarrow{d} \\
\overrightarrow{g} & \overrightarrow{c}
 \end{array} 
\Bigg|
u_1
\Bigg|
u_3
\Bigg)
W \Bigg(
 \begin{array}{cc}
\overrightarrow{g} & \overrightarrow{c} \\
\overrightarrow{a} & \overrightarrow{b}
 \end{array} 
\Bigg|
u_2
\Bigg|
u_3
\Bigg)
W \Bigg(
 \begin{array}{cc}
\overrightarrow{e} & \overrightarrow{g} \\
\overrightarrow{f} & \overrightarrow{a}
 \end{array} 
\Bigg|
u_2
\Bigg|
u_1
\Bigg). \label{yangbaxter}
\end{align}

One next introduces
the product of weights (Figure \ref{pictureprodcutofweights})
as
\begin{align}
&T_{M}^j \Bigg(
 \begin{array}{cccc}
\overrightarrow{a_{0}^j} & \overrightarrow{a_1^j} & \cdots &
\overrightarrow{a_M^j} \\
\overrightarrow{b_{0}^j} & \overrightarrow{b_1^j} & \cdots &
\overrightarrow{b_M^j}
 \end{array} 
\Bigg|
u
\Bigg|
v_1,\dots,v_M
\Bigg) \nonumber \\
=&
W \Bigg(
 \begin{array}{cc}
\overrightarrow{a_0^j} & \overrightarrow{a_1^j} \\
\overrightarrow{b_0^j} & \overrightarrow{b_1^j}
 \end{array} 
\Bigg|
u
\Bigg|
v_1
\Bigg)
W \Bigg(
 \begin{array}{cc}
\overrightarrow{a_1^j} & \overrightarrow{a_2^j} \\
\overrightarrow{b_1^j} & \overrightarrow{b_2^j}
 \end{array} 
\Bigg|
u
\Bigg|
v_2
\Bigg)
\cdots
W \Bigg(
 \begin{array}{cc}
\overrightarrow{a_{M-1}^j} & \overrightarrow{a_M^j} \\
\overrightarrow{b_{M-1}^j} & \overrightarrow{b_M^j}
 \end{array} 
\Bigg|
u
\Bigg|
v_M
\Bigg). \label{productweights}
\end{align}

We are now in a position to define the partition functions
of the Deguchi-Martin model which we investigate in this paper.

\begin{definition} \label{definitionofwavefunctions}

We define the partition functions of the elliptic supersymmetric
Deguchi-Martin model using the product of weights
\eqref{productweights} as
\begin{align}
&W_{M,N}(u_1,\dots,u_N|v_1,\dots,v_M|x_1,\dots,x_N|a_{12})
=\sum_{\{ \overrightarrow{b^0} \},\{ \overrightarrow{b^1} \},\dots,\{ \overrightarrow{b^{N-2}} \}} \nonumber \\
&T_{M}^0 \Bigg(
 \begin{array}{ccccc}
\overrightarrow{a_{12}} & \overrightarrow{a_{12}}+\hat{e}_1 
& \cdots &
\overrightarrow{a_{12}}+(M-1)\hat{e}_1 
&
\overrightarrow{a_{12}}+M\hat{e}_1 \\
\overrightarrow{a_{12}}+\hat{e}_1 & \overrightarrow{b_1^0}
& \cdots &
\overrightarrow{b_{M-1}^0} 
&
\overrightarrow{a_{12}}+M\hat{e}_1+\hat{e}_2
 \end{array} 
\Bigg|
u_N
\Bigg|
v_1,\dots,v_M
\Bigg) \nonumber \\
\times \prod_{j=1}^{N-2}
&T_{M}^j \Bigg(
 \begin{array}{ccccc}
\overrightarrow{a_{12}}+j\hat{e}_1 & \overrightarrow{b_{1}^{j-1}}
& \cdots &
\overrightarrow{b_{M-1}^{j-1}}
&
\overrightarrow{a_{12}}+M\hat{e}_1+j\hat{e}_2 \\
\overrightarrow{a_{12}}+(j+1)\hat{e}_1 & \overrightarrow{b_1^j}
& \cdots &
\overrightarrow{b_{M-1}^j}
&
\overrightarrow{a_{12}}+M\hat{e}_1+(j+1)\hat{e}_2
 \end{array} 
\Bigg|
u_{N-j}
\Bigg|
v_1,\dots,v_M
\Bigg) \nonumber \\
\times
&T_{M}^{N-1} \Bigg(
 \begin{array}{ccccc}
\overrightarrow{a_{12}}+(N-1)\hat{e}_1 & \overrightarrow{b_{1}^{N-2}}
& \cdots &
\overrightarrow{b_{M-1}^{N-2}}
&
\overrightarrow{a_{12}}+M\hat{e}_1+(N-1)\hat{e}_2 \\
\overrightarrow{c_{0}^{N-1}} & \overrightarrow{c_{1}^{N-1}}
& \cdots &
\overrightarrow{c_{M-1}^{N-1}}
&
\overrightarrow{c_M^{N-1}}
 \end{array} 
\Bigg|
u_1
\Bigg|
v_1,\dots,v_M
\Bigg), \label{wavefunction}
\end{align}
where the state vectors
$\overrightarrow{c_0^{N-1}},\overrightarrow{c_1^{N-1}},\dots,
\overrightarrow{c_{M-1}^{N-1}},\overrightarrow{c_{M}^{N-1}}$
are fixed by the sequence of integers $1 \le x_1 < x_2 < \cdots < x_N \le M$
as
\begin{align}
&\overrightarrow{c_0^{N-1}},\overrightarrow{c_1^{N-1}},\dots,
\overrightarrow{c_{M-1}^{N-1}},\overrightarrow{c_{M}^{N-1}} \nonumber \\
=&\overrightarrow{a_{12}}+N\hat{e}_1,\dots,
\overrightarrow{a_{12}}+(N+x_1-1)\hat{e}_1,
\overrightarrow{a_{12}}+(N+x_1-1)\hat{e}_1+\hat{e}_2, \nonumber \\
&\overrightarrow{a_{12}}+(N+x_1)\hat{e}_1+\hat{e}_2,\dots,
\overrightarrow{a_{12}}+(N+x_1+x_2-3)\hat{e}_1+\hat{e}_2,
\overrightarrow{a_{12}}+(N+x_1+x_2-3)\hat{e}_1+2\hat{e}_2, \nonumber \\
&\overrightarrow{a_{12}}+(N+x_1+x_2-2)\hat{e}_1+2\hat{e}_2,\dots,
\overrightarrow{a_{12}}+M\hat{e}_1+N\hat{e}_2.
\end{align}
The sequences of integers $1 \le x_1 < x_2 < \cdots < x_N \le M$
label the positions of bottom edges
of the partition functions where the difference
of the state vectors of adjacent vertices differ by $\hat{e}_2$.

The sum $\displaystyle \sum_{\{ \overrightarrow{b^0} \},\{ \overrightarrow{b^1} \},\dots,\{ \overrightarrow{b^{N-2}} \}}$ in \eqref{wavefunction}
means that we take the sum over all inner state vectors
$\overrightarrow{b_j^0}, \overrightarrow{b_j^1},\dots, \overrightarrow{b_j^{N-2}}$, $j=1,\dots,M-1$.
We remark that from the definition of the weights
\eqref{weightsone},\eqref{weightstwo},\eqref{weightsthree},
the dependence on the state vector $\overrightarrow{a_{12}}=(a_1,a_2)$
in the top left corner vertex is reflected to the partition functions
\eqref{wavefunction} in the form of the scalar quantity
$a_{12}=a_1+a_2+\omega_{12}$, hence we write the partition functions as
$W_{M,N}(u_1,\dots,u_N|v_1,\dots,v_M|x_1,\dots,x_N|a_{12})$.

\end{definition}

\begin{figure}[ht]
\includegraphics[width=12cm]{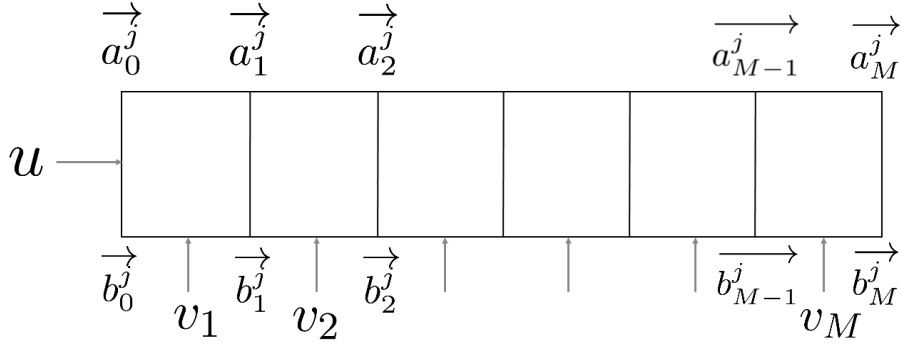}
\caption{The product of weights
\eqref{productweights} associated with one row configurations.
}
\label{pictureprodcutofweights}
\end{figure}

\begin{figure}[ht]
\includegraphics[width=12cm]{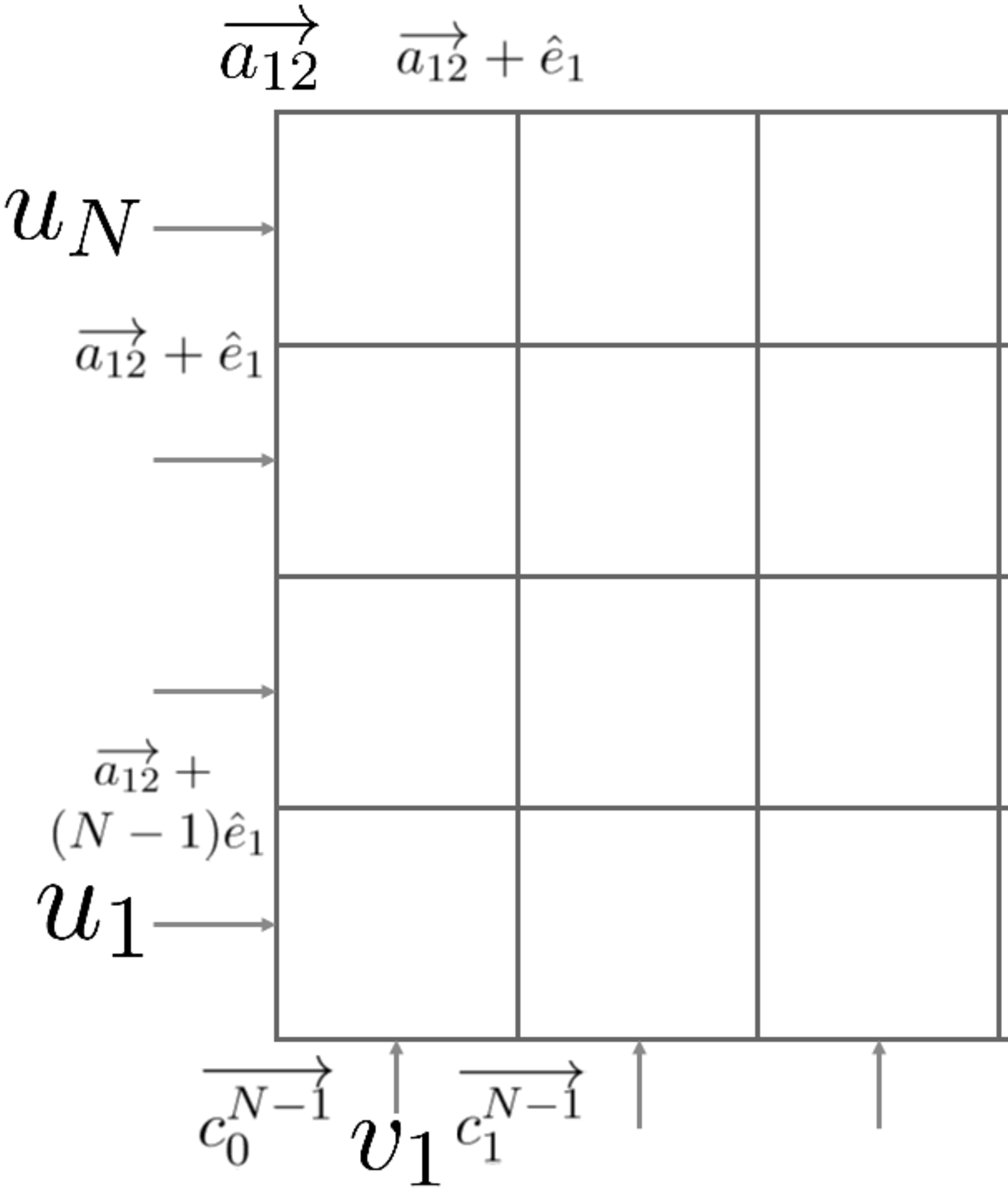}
\caption{The partition functions
$W_{M,N}(u_1,\dots,u_N|v_1,\dots,v_M|x_1,\dots,x_N|a_{12})$
\eqref{wavefunction}.
}
\label{picturewavefunctions}
\end{figure}

\begin{figure}[ht]
\includegraphics[width=12cm]{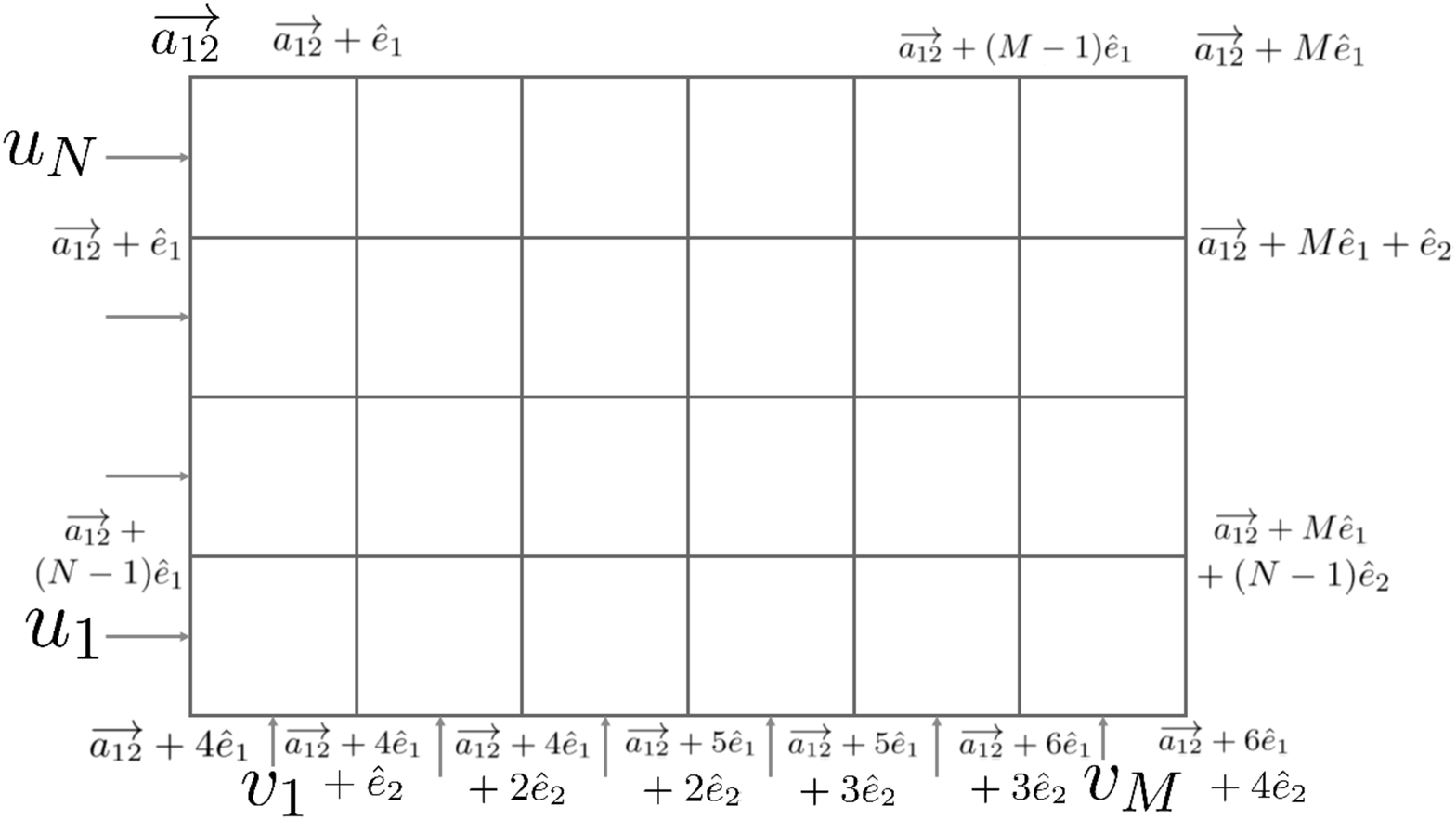}
\caption{An example
$W_{6,4}(u_1,\dots,u_4|v_1,\dots,v_6|1,2,4,6|a_{12})$
of the partition functions
\eqref{wavefunction}.
Note that $x_1=1$, $x_2=2$, $x_3=4$, $x_4=6$
labels the positions of the edges of the bottom part of the partition functions
where the difference of the adjacent state vectors of the vertices
is $\hat{e}_2$.
}
\label{picturewavefunctionsexample}
\end{figure}

\begin{figure}[ht]
\includegraphics[width=12cm]{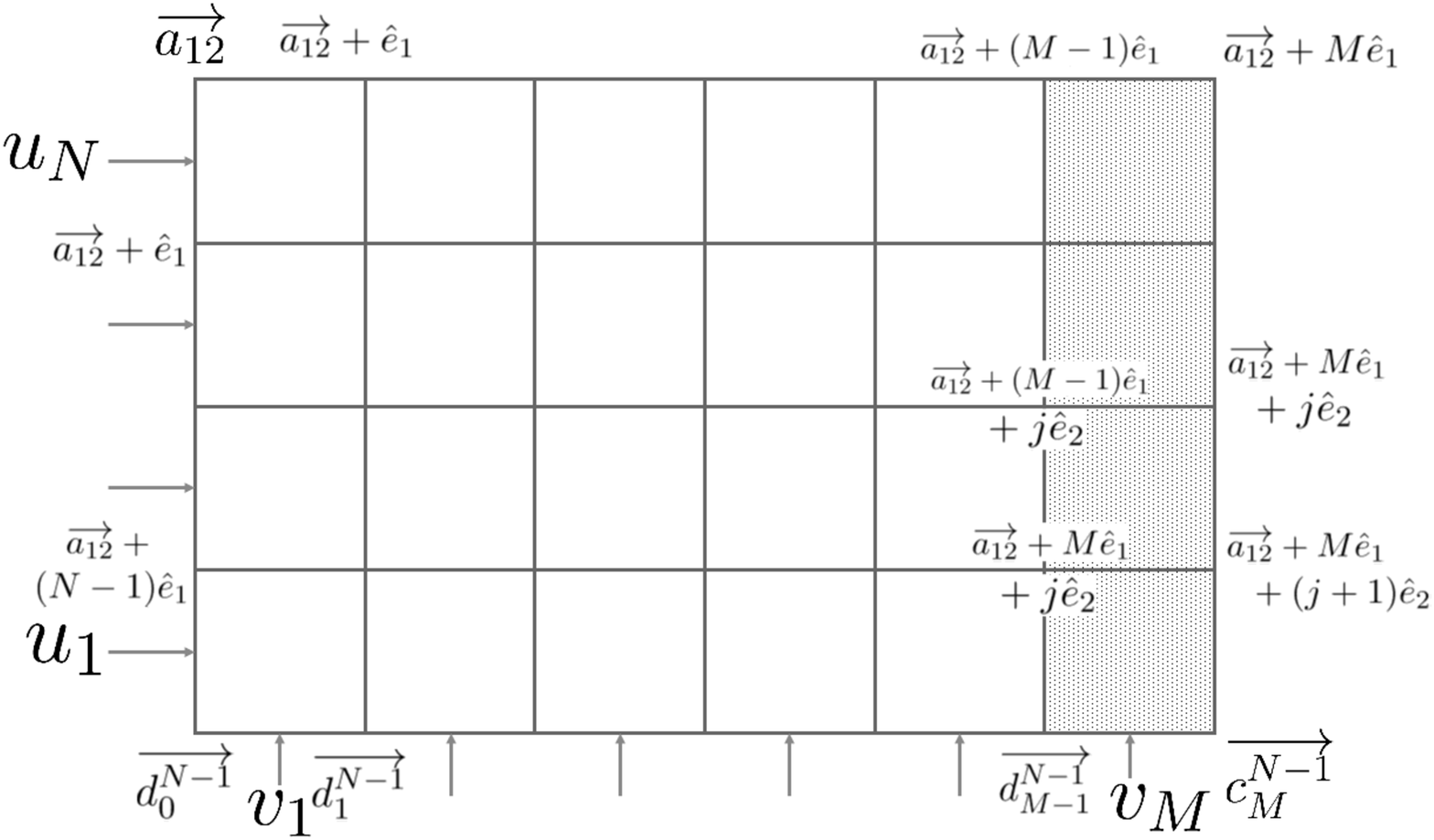}
\caption{A graphical description of a summand
in the decompostion of the partition functions
\eqref{separation}.
The unshaded part and the shaded part
corresponds to
$P_{j}(u_1,\dots,u_N|v_1,\dots,v_{M-1}|x_1,\dots,x_{N-1}
|a_{12})$ \eqref{separatedparttwo}
and
$
\widetilde{W}_j(u_1,\dots,u_N|v_M|a_{12})$
\eqref{separatedpartone}
,
respectively.
}
\label{picturedecomposition}
\end{figure}

\begin{figure}[ht]
\includegraphics[width=12cm]{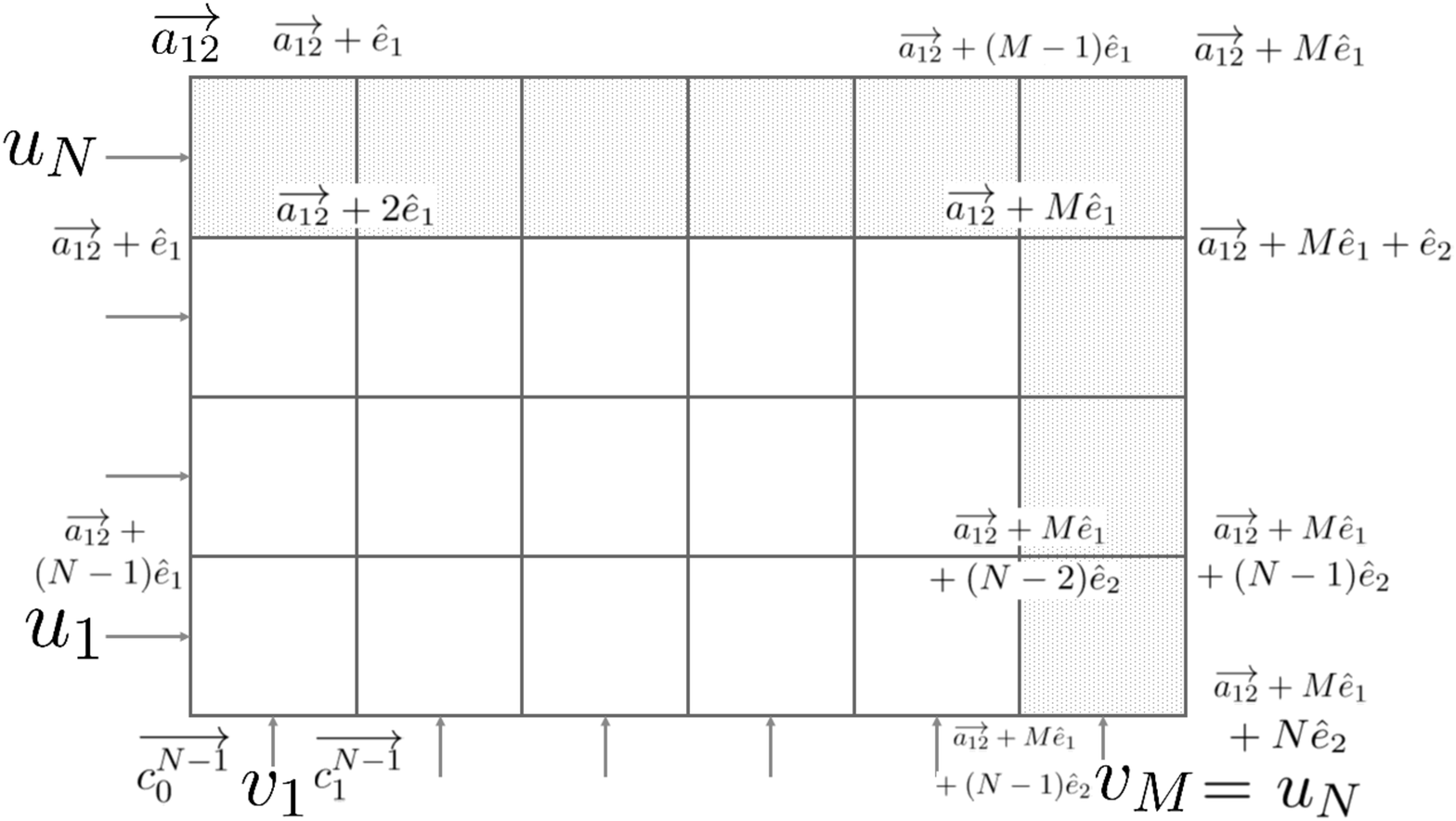}
\caption{The recursion relation
$W_{M,N}(u_1,\dots,u_N|v_1,\dots,v_M|x_1,\dots,x_N|a_{12})$,
$x_N=M$ evaluated at $v_M=u_N$
\eqref{ordinaryrecursionwavefunction}
.
The states of the vertices around the shaded faces are freezed.
}
\label{ordinarypicturerecursion}
\end{figure}

\begin{figure}[ht]
\includegraphics[width=12cm]{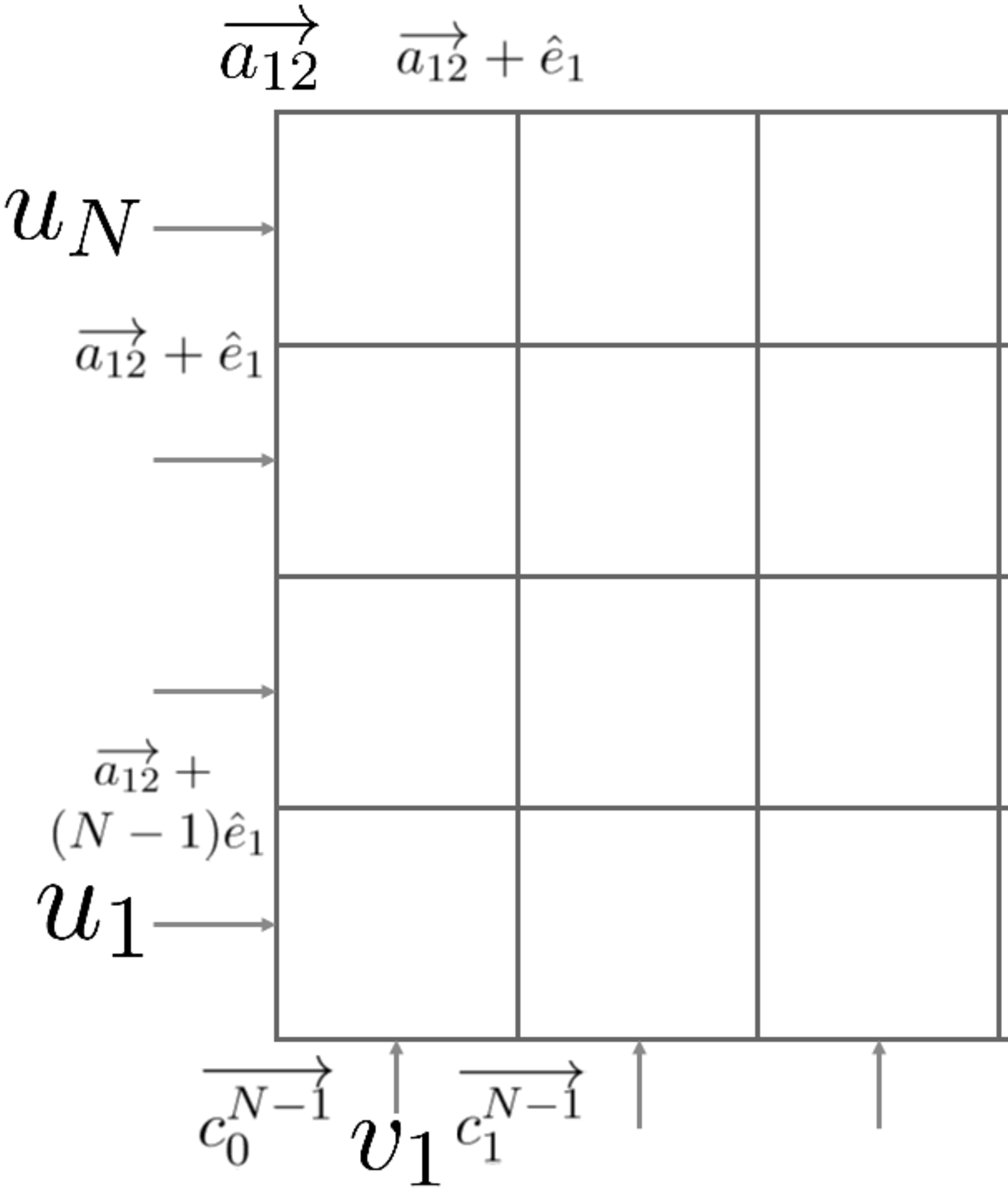}
\caption{The factorization of
$W_{M,N}(u_1,\dots,u_N|v_1,\dots,v_M|x_1,\dots,x_N|a_{12})$,
$x_N \neq M$
\eqref{ordinaryrecursionwavefunction2}
.
The states of the vertices around the shaded faces are freezed.
}
\label{ordinarypicturerecursion2}
\end{figure}

See Figures \ref{picturewavefunctions} and \ref{picturewavefunctionsexample}
for pictorial descriptions
of the definition and an example of the partition functions
\eqref{wavefunction}.
In the next section, we examine the properties of the 
partition functions.

\section{Izergin-Korepin analysis}

In this section, we determine the properties of the partition functions
of the elliptic supersymmetric face model
$W_{M,N}(u_1,\dots,u_N|v_1,\dots,v_M|x_1,\dots,x_N|a_{12})$
by extending the Izergin-Korepin analysis
on the domain wall boundary partition functions
\cite{Ko,Iz}.

\begin{proposition} 
\label{ordinarypropertiesfordomainwallboundarypartitionfunction}
The partition functions
$W_{M,N}(u_1,\dots,u_N|v_1,\dots,v_M|x_1,\dots,x_N|a_{12})$
satisfies the following properties. \\
\\
 (1) When $x_N=M$, the partition functions
$W_{M,N}(u_1,\dots,u_N|v_1,\dots,v_M|x_1,\dots,x_N|a_{12})$,
regarded as a polynomial in $\displaystyle y_M:=\frac{\lambda}{2K_1}v_M$,
is an elliptic polynomial in $\Theta_N(\chi)$.
\\
 (2) The partition functions $W_{M,N}(u_{\sigma(1)},\dots,u_{\sigma(N)}|v_1,\dots,v_M|x_1,\dots,x_N|a_{12})$ with the ordering of the spectral parameters permuted
$u_{\sigma(1)}, \dots, u_{\sigma(N)}$, $\sigma \in S_N$ are related with
the unpermuted one
$W_{M,N}(u_1,\dots,u_N|v_1,\dots,v_M|x_1,\dots,x_N|a_{12})$ by
the following relation
\begin{align}
&\prod_{\substack{1 \le j < k \le N \\ \sigma(j) > \sigma(k)}}
[1+u_{\sigma(k)}-u_{\sigma(j)}]
W_{M,N}(u_1,\dots,u_N|v_1,\dots,v_M|x_1,\dots,x_N|a_{12}) \nonumber \\
=&
\prod_{\substack{1 \le j < k \le N \\ \sigma(j) > \sigma(k)}}
[1+u_{\sigma(j)}-u_{\sigma(k)}]
W_{M,N}(u_{\sigma(1)},\dots,u_{\sigma(N)}|v_1,\dots,v_M|x_1,\dots,x_N|a_{12}).
\label{permutationwavefunction}
\end{align}
\\
(3) The following recursive relations between the
partition functions hold if $x_N=M$
(Figure \ref{ordinarypicturerecursion}):
\begin{align}
&W_{M,N}(u_1,\dots,u_N|v_1,\dots,v_M|x_1,\dots,x_N|a_{12})
|_{v_M=u_N}
\nonumber \\
=&\prod_{j=1}^{N-1} \frac{[1-u_j+u_N]}{[1]}
\prod_{j=1}^{M-1} \frac{[1+u_N-v_j]}{[1]}
\nonumber \\
&\times W_{M-1,N-1}(u_1,\dots,u_{N-1}|v_1,\dots,v_{M-1}|x_1,\dots,x_{N-1}
|a_{12}+1)
. \label{ordinaryrecursionwavefunction}
\end{align}

If $x_N \neq M$, the following factorizations hold for the partition functions
(Figure \ref{ordinarypicturerecursion2}):
\begin{align}
&W_{M,N}(u_1,\dots,u_N|v_1,\dots,v_M|x_1,\dots,x_N|a_{12})
 \nonumber \\
=&\frac{[a_{21}-M-N+1]}{[a_{21}-M+1]}
\prod_{j=1}^N \frac{[u_j-v_M]}{[1]}
W_{M-1,N}(u_1,\dots,u_N|v_1,\dots,v_{M-1}|x_1,\dots,x_N|a_{12}).
\label{ordinaryrecursionwavefunction2}
\end{align}
\\
(4) The following holds for the case $N=1$, $x_1=M$
\begin{align}
&
W_{M,1}(u|v_1,\dots,v_M|M|a_{12})=
\frac{[-u+v_M+a_{12}+M-1]}{[a_{12}+M-1]}
\prod_{k=1}^{M-1} \frac{[1+u-v_k]}{[1]}.
\label{ordinaryinitialrecursion}
\end{align}
\end{proposition}

\begin{proof}
Properties (1), (2) and (3) for the case $x_N=M$
can be proved essentially in the same way with
the domain wall boundary partition functions,
which is given in the PhD thesis of Zuparic \cite{Zuparic}
(note there is a small difference in the setting.
For example, the ordering of the spectral parameters is inverted).
We just have to treat the case $x_N \neq M$ separately,
which is given as the factorization formula of the
partition functions in Property (3).

To show Property (1),
we first decompose the partition functions satisfying $x_N=M$
in the following form as we do for the case of vertex models.
We split the partition functions as a sum of the products
of the partition functions $P_j(u_1,\dots,u_N|v_1,\dots,v_{M-1}|x_1,\dots,x_{N-1}|a_{12})$ and one-column partition functions
$\widetilde{W}_j(u_1,\dots,u_N|v_M|a_{12})$
(see Figure \ref{picturedecomposition} for this decomposition)
\begin{align}
&W_{M,N}(u_1,\dots,u_N|v_1,\dots,v_M|x_1,\dots,x_N|a_{12})
 \nonumber \\
=&\sum_{j=0}^{N-1} P_{j}(u_1,\dots,u_N|v_1,\dots,v_{M-1}|x_1,\dots,x_{N-1}
|a_{12})
\widetilde{W}_j(u_1,\dots,u_N|v_M|a_{12}). \label{separation}
\end{align}
Here, $\widetilde{W}_j(u_1,\dots,u_N|v_M|a_{12})$
are the one-column partition functions
corresponding to the shaded part in Figure \ref{picturedecomposition}
\begin{align}
&\widetilde{W}_j(u_1,\dots,u_N|v_M|a_{12}) \nonumber \\
=&\prod_{k=0}^{j-1}
W \Bigg(
 \begin{array}{cc}
\overrightarrow{a_{12}}+(M-1)\hat{e}_1+k\hat{e}_2 & \overrightarrow{a_{12}}+M\hat{e}_1+k\hat{e}_2 \\
\overrightarrow{a_{12}}+(M-1)\hat{e}_1+(k+1)\hat{e}_2 & \overrightarrow{a_{12}}+M\hat{e}_1+(k+1)\hat{e}_2
 \end{array} 
\Bigg|
u_{N-k}
\Bigg|
v_M
\Bigg) \nonumber \\
&\times
W \Bigg(
 \begin{array}{cc}
\overrightarrow{a_{12}}+(M-1)\hat{e}_1+j\hat{e}_2 & \overrightarrow{a_{12}}+M\hat{e}_1+j\hat{e}_2 \\
\overrightarrow{a_{12}}+M\hat{e}_1+j\hat{e}_2 & \overrightarrow{a_{12}}+M\hat{e}_1+(j+1)\hat{e}_2
 \end{array} 
\Bigg|
u_{N-j}
\Bigg|
v_M
\Bigg) \nonumber \\
&\times \prod_{k=j}^{N-2}
W \Bigg(
 \begin{array}{cc}
\overrightarrow{a_{12}}+M\hat{e}_1+k\hat{e}_2 & \overrightarrow{a_{12}}+M\hat{e}_1+(k+1)\hat{e}_2 \\
\overrightarrow{a_{12}}+M\hat{e}_1+(k+1)\hat{e}_2 & \overrightarrow{a_{12}}+M\hat{e}_1+(k+2)\hat{e}_2
 \end{array} 
\Bigg|
u_{N-k-1}
\Bigg|
v_M
\Bigg). \label{separatedpartone}
\end{align}

The partition functions
$P_j(u_1,\dots,u_N|v_1,\dots,v_{M-1}|x_1,\dots,x_{N-1}|a_{12})$
corresponding to the unshaded part in Figure \ref{picturedecomposition}
is explicitly written as
\begin{align}
&P_j(u_1,\dots,u_N|v_1,\dots,v_{M-1}|x_1,\dots,x_{N-1}|a_{12})
=\sum_{\{ \overrightarrow{b^0} \},\{ \overrightarrow{b^1} \},\dots,\{ \overrightarrow{b^{N-2}} \}} \nonumber \\
&T_{M-1}^0 \Bigg(
 \begin{array}{ccccc}
\overrightarrow{a_{12}} & \overrightarrow{a_{12}}+\hat{e}_1 
& \cdots &
\overrightarrow{a_{12}}+(M-2)\hat{e}_1 
&
\overrightarrow{a_{12}}+(M-1)\hat{e}_1 \\
\overrightarrow{a_{12}}+\hat{e}_1 & \overrightarrow{b_1^0}
& \cdots &
\overrightarrow{b_{M-2}^0} 
&
\overrightarrow{a_{12}}+(M-1)\hat{e}_1+\hat{e}_2
 \end{array} 
\Bigg|
u_N
\Bigg) \nonumber \\
\times \prod_{k=1}^{j-1}
&T_{M-1}^k \Bigg(
 \begin{array}{ccccc}
\overrightarrow{a_{12}}+k\hat{e}_1 & \overrightarrow{b_{1}^{k-1}}
& \cdots &
\overrightarrow{b_{M-2}^{k-1}}
&
\overrightarrow{a_{12}}+(M-1)\hat{e}_1+k\hat{e}_2 \\
\overrightarrow{a_{12}}+(k+1)\hat{e}_1 & \overrightarrow{b_1^k}
& \cdots &
\overrightarrow{b_{M-2}^k}
&
\overrightarrow{a_{12}}+(M-1)\hat{e}_1+(k+1)\hat{e}_2
 \end{array} 
\Bigg|
u_{N-k}
\Bigg) \nonumber \\
\times&T_{M-1}^j \Bigg(
 \begin{array}{ccccc}
\overrightarrow{a_{12}}+j\hat{e}_1 & \overrightarrow{b_{1}^{j-1}}
& \cdots &
\overrightarrow{b_{M-2}^{j-1}}
&
\overrightarrow{a_{12}}+(M-1)\hat{e}_1+j\hat{e}_2 \\
\overrightarrow{a_{12}}+(j+1)\hat{e}_1 & \overrightarrow{b_1^j}
& \cdots &
\overrightarrow{b_{M-2}^j}
&
\overrightarrow{a_{12}}+M\hat{e}_1+j\hat{e}_2
 \end{array} 
\Bigg|
u_{N-j}
\Bigg) \nonumber \\
\times \prod_{k=j+1}^{N-2}
&T_{M-1}^k \Bigg(
 \begin{array}{ccccc}
\overrightarrow{a_{12}}+k\hat{e}_1 & \overrightarrow{b_{1}^{k-1}}
& \cdots &
\overrightarrow{b_{M-2}^{k-1}}
&
\overrightarrow{a_{12}}+M\hat{e}_1+(k-1)\hat{e}_2 \\
\overrightarrow{a_{12}}+(k+1)\hat{e}_1 & \overrightarrow{b_1^k}
& \cdots &
\overrightarrow{b_{M-2}^k}
&
\overrightarrow{a_{12}}+M\hat{e}_1+k\hat{e}_2
 \end{array} 
\Bigg|
u_{N-k}
\Bigg) \nonumber \\
\times
&T_{M-1}^{N-1} \Bigg(
 \begin{array}{ccccc}
\overrightarrow{a_{12}}+(N-1)\hat{e}_1 & \overrightarrow{b_{1}^{N-2}}
& \cdots &
\overrightarrow{b_{M-2}^{N-2}}
&
\overrightarrow{a_{12}}+M\hat{e}_1+(N-2)\hat{e}_2 \\
\overrightarrow{d_{0}^{N-1}} & \overrightarrow{d_{1}^{N-1}}
& \cdots &
\overrightarrow{d_{M-2}^{N-1}}
&
\overrightarrow{d_{M-1}^{N-1}}
 \end{array} 
\Bigg|
u_1
\Bigg), \label{separatedparttwo}
\end{align}
where the state vectors
$\overrightarrow{d_0^{N-1}},\overrightarrow{d_1^{N-1}},\dots,
\overrightarrow{d_{M-2}^{N-1}},\overrightarrow{d_{M-1}^{N-1}}$
are fixed using the sequence of integers $1 \le x_1 < x_2 < \cdots < x_{N-1} \le M-1$
as
\begin{align}
&\overrightarrow{d_0^{N-1}},\overrightarrow{d_1^{N-1}},\dots,
\overrightarrow{d_{M-2}^{N-1}},\overrightarrow{d_{M-1}^{N-1}} \nonumber \\
=&\overrightarrow{a_{12}}+N\hat{e}_1,\dots,
\overrightarrow{a_{12}}+(N+x_1-1)\hat{e}_1,
\overrightarrow{a_{12}}+(N+x_1-1)\hat{e}_1+\hat{e}_2, \nonumber \\
&\overrightarrow{a_{12}}+(N+x_1)\hat{e}_1+\hat{e}_2,\dots,
\overrightarrow{a_{12}}+(N+x_1+x_2-3)\hat{e}_1+\hat{e}_2,
\overrightarrow{a_{12}}+(N+x_1+x_2-3)\hat{e}_1+2\hat{e}_2, \nonumber \\
&\overrightarrow{a_{12}}+(N+x_1+x_2-2)\hat{e}_1+2\hat{e}_2,\dots,
\overrightarrow{a_{12}}+M\hat{e}_1+(N-1)\hat{e}_2.
\end{align}
We have omitted writing the dependence on the inhomogeneous parameters
in $T_{M-1}^k$
in \eqref{separatedparttwo}.
It depends on $v_1,\dots,v_{M-1}$, but not on $v_M$.
This means that
\eqref{separatedparttwo} does not depend on the inhomogeneous parameter $v_M$,
so we only have to examine the other parts \eqref{separatedpartone}
if one analyzes the partition functions as a function of $v_M$.
From the definition of the weights
\eqref{weightsone},
\eqref{weightstwo},
\eqref{weightsthree},
$\widetilde{W}_j(u_1,\dots,u_N|v_M|a_{12})$ is explicitly calculated as
\begin{align}
&\widetilde{W}_j(u_1,\dots,u_N|v_M|a_{12})
= \frac{[a_{12}+M-1+j-u_{N-j}+v_M]}{[a_{12}+M-1+j]} \nonumber \\
\times&\prod_{k=N-j+1}^{N} \frac{[u_k-v_M][a_{21}-M+k-N]}{[1][a_{21}-M+1+k-N]}
\prod_{k=1}^{N-j-1} \frac{[1-u_k+v_M]}{[1]}. \label{explicitforuse}
\end{align}

We then view $\widetilde{W}_j(u_1,\dots,u_N|v_M|a_{12})$
as a function of $v_M$.
From its explicit form \eqref{explicitforuse}
and using \eqref{qpuseone} and \eqref{qpusetwo},
$a_{12}=-a_{21}$ and $[-u]=-[u]$,
one easily calculates the quasi-periodicities
of the function $\widetilde{W}_j(u_1,\dots,u_N|v_M|a_{12})$
\begin{align}
&\widetilde{W}_j(u_1,\dots,u_N|v_M+2K_1/\lambda|a_{12}) \nonumber \\
=&(-1)^N \widetilde{W}_j(u_1,\dots,u_N|v_M|a_{12}), \label{qpfirst} \\
&\widetilde{W}_j(u_1,\dots,u_N|v_M+2 i K_2/\lambda|a_{12}) \nonumber \\
=&(-q^{-1})^N
\mathrm{exp} \Bigg(-\frac{i \pi\lambda}{K_1} 
\Bigg(Nv_M-\sum_{j=1}^N u_j+a_{12}+N+M-2 \Bigg)
\Bigg)
\widetilde{W}_j(u_1,\dots,u_N|v_M|a_{12}). \label{qpsecond}
\end{align}
From \eqref{qpfirst}, \eqref{qpsecond} and \eqref{separation},
one finds the same quasi-periodicities as above also hold for
the partition functions
\begin{align}
&W_{M,N}(u_1,\dots,u_N|v_1,\dots,v_M+2K_1/\lambda|x_1,\dots,x_N|a_{12})
\nonumber \\
=&(-1)^N W_{M,N}(u_1,\dots,u_N|v_1,\dots,v_M|x_1,\dots,x_N|a_{12}),
\label{wavefunctionqpfirst} \\
&
W_{M,N}(u_1,\dots,u_N|v_1,\dots,v_M+2 i K_2/\lambda|x_1,\dots,x_N|a_{12})
\nonumber \\
=&(-q^{-1})^N
\mathrm{exp} \Bigg(-\frac{i \pi\lambda}{K_1} 
\Bigg(Nv_M-\sum_{j=1}^N u_j+a_{12}+N+M-2 \Bigg)
\Bigg) \nonumber \\
&\times W_{M,N}(u_1,\dots,u_N|v_1,\dots,v_M|x_1,\dots,x_N|a_{12})
. \label{wavefunctionqpsecond}
\end{align}
If we rescale the parameter from $v_M$ to
$\displaystyle y_M:=\frac{\lambda}{2K_1} v_M$,
and define $\phi(y_M)$ as
\begin{align}
\phi(y_M)=W_{M,N}(u_1,\dots,u_N|v_1,\dots,v_M|x_1,\dots,x_N|a_{12}),
\end{align}
the quasi-periodicities \eqref{wavefunctionqpfirst} and \eqref{wavefunctionqpsecond} for $W_{M,N}(u_1,\dots,u_N|v_1,\dots,v_M|x_1,\dots,x_N|a_{12})$
can be rewritten as a function $\phi(y_M)$ of $y_M$ as
\begin{align}
\phi(y_M+1)&=(-1)^N \phi(y_M), \\
\phi(y_M+\tau)&=(-1)^N \mathrm{exp} \Bigg( \frac{\pi \lambda \tau}{K_2}
\Bigg(
\sum_{j=1}^N u_j-a_{12}-N-M+2
\Bigg)
\Bigg)
e^{-2\pi i Ny_M-\pi i N \tau} \phi(y_M),
\end{align}
where the nome $\tau$ is given by $\displaystyle \tau=i \frac{K_2}{K_1}$.
These are exactly the quasi-periodicities
\eqref{propertyuseone}, \eqref{propertyusetwo} with characters
$\chi(1)=(-1)^N$, $\displaystyle \chi(\tau)=(-1)^N \mathrm{exp} 
\Bigg( \frac{\pi\lambda \tau}{K_2}
\Bigg(
\sum_{j=1}^N u_j-a_{12}-N-M+2
\Bigg)
\Bigg)$.
Hence, it is an elliptic polynomial in $\Theta_N(\chi)$.

Property (2) can also be shown in a standard way.
By attaching a face having the weight \\
$
W \Bigg(
 \begin{array}{cc}
\overrightarrow{a_{12}}+M\hat{e}_1+(N-j-1)\hat{e}_2 & \overrightarrow{a_{12}}+M\hat{e}_1+(N-j)\hat{e}_2 \\
\overrightarrow{a_{12}}+M\hat{e}_1+(N-j)\hat{e}_2 & \overrightarrow{a_{12}}
+M\hat{e}_1+(N-j+1)\hat{e}_2
 \end{array} 
\Bigg|
u_{j+1}
\Bigg|
u_j
\Bigg)
$
on the right side of the partition functions
$W_{M,N}(u_1,\dots,u_N|v_1,\dots,v_M|x_1,\dots,x_N|a_{12})$
and using the face-type Yang-Baxter relation \eqref{yangbaxter}
repeatedly,
a face having the weight \\
$
W \Bigg(
 \begin{array}{cc}
\overrightarrow{a_{12}}+(N-j-1)\hat{e}_1 & \overrightarrow{a_{12}}+(N-j)\hat{e}_1 \\
\overrightarrow{a_{12}}+(N-j)\hat{e}_1 & \overrightarrow{a_{12}}
+(N-j+1)\hat{e}_1
 \end{array} 
\Bigg|
u_{j+1}
\Bigg|
u_j
\Bigg)
$
gets out of the left side of the partition functions.
The ordering of the spectral parameters of the partition functions
is now changed to
$W_{M,N}(u_1,\dots,u_{j+1},u_j,\dots,u_N|v_1,\dots,v_M|x_1,\dots,x_N|a_{12})$,
and we get
\begin{align}
&W \Bigg(
 \begin{array}{cc}
\overrightarrow{a_{12}}+M\hat{e}_1+(N-j-1)\hat{e}_2 & \overrightarrow{a_{12}}+M\hat{e}_1+(N-j)\hat{e}_2 \\
\overrightarrow{a_{12}}+M\hat{e}_1+(N-j)\hat{e}_2 & \overrightarrow{a_{12}}
+M\hat{e}_1+(N-j+1)\hat{e}_2
 \end{array} 
\Bigg|
u_{j+1}
\Bigg|
u_j
\Bigg) \nonumber \\
\times&W_{M,N}(u_1,\dots,u_N|v_1,\dots,v_M|x_1,\dots,x_N|a_{12})
\nonumber \\
=&
W \Bigg(
 \begin{array}{cc}
\overrightarrow{a_{12}}+(N-j-1)\hat{e}_1 & \overrightarrow{a_{12}}+(N-j)\hat{e}_1 \\
\overrightarrow{a_{12}}+(N-j)\hat{e}_1 & \overrightarrow{a_{12}}
+(N-j+1)\hat{e}_1
 \end{array} 
\Bigg|
u_{j+1}
\Bigg|
u_j
\Bigg)
\nonumber \\
\times&W_{M,N}(u_1,\dots,u_{j+1},u_j,\dots,u_N|v_1,\dots,v_M|x_1,\dots,x_N|a_{12}).
\end{align}
Simplifying the above equality, we have
\begin{align}
&[1+u_{j}-u_{j+1}]W_{M,N}(u_1,\dots,u_N|v_1,\dots,v_M|x_1,\dots,x_N|a_{12})
\nonumber \\
=&
[1+u_{j+1}-u_j]W_{M,N}(u_1,\dots,u_{j+1},u_j,\dots,u_N|v_1,\dots,v_M|x_1,\dots,x_N|a_{12}). \label{basiccommutation}
\end{align}
The general case
\eqref{permutationwavefunction} can be obtained by
using the commutation relation \eqref{basiccommutation} repeatedly.
For example, the case $N=3$, $\sigma=(3,1,2)$ can be shown as follows:
\begin{align}
&W_{M,3}(u_1,u_2,u_3|v_1,\dots,v_M|x_1,x_2,x_3|a_{12}) \nonumber \\
=&\frac{[1+u_3-u_2]}{[1+u_2-u_3]}
W_{M,3}(u_1,u_3,u_2|v_1,\dots,v_M|x_1,x_2,x_3|a_{12}) \nonumber \\
=&\frac{[1+u_3-u_2]}{[1+u_2-u_3]}
\frac{[1+u_3-u_1]}{[1+u_1-u_3]}
W_{M,3}(u_3,u_1,u_2|v_1,\dots,v_M|x_1,x_2,x_3|a_{12})
.
\end{align}

Let us show Property (3) for the case $x_N = M$.
Before showing this,
we remark that from Property (1),
the partition functions
$W_{M,N}(u_1,\dots,u_N|v_1,\dots,v_M|x_1,\dots,x_N|a_{12})$
is an elliptic polynomial of degree $N$ in $v_M$,
and the properties of the elliptic polynomials listed in
Proposition \ref{propositionelliptic}
allow us to use the Izergin-Korepin trick.
If one relates the evaluations of
the partition functions at $N$ distinct points in $v_M$
with smaller ones, and if one can find the functions
satisfying the same recursion, then the Izergin-Korepin method
is successful. One should keep in mind that
for the case of the partition functions, one has another
case $x_N \neq M$, which must be taken into account.
We also need Property (4) to assure the uniqueness,
since this property corresponds to the initial condition
of the recursion relation.
\eqref{ordinaryrecursionwavefunction} is the result of the
evaluation of the partition functions
at the point $v_M=u_N$.
The other $(N-1)$ points $v_M=u_j$, $j=1,\dots,N-1$
can be evaluated using Property (2),
hence if one shows that certain explicit functions satisfy Property (2)
and the evaluation at $v_M=u_N$,
the proof of the recursion relation for the case $x_N=M$ is succeeded.
This will be performed in section 5.

One can show the recursion relation
\eqref{ordinaryrecursionwavefunction}
by using the graphical
representation of the partition functions $W_{M,N}(u_1,\dots,u_N|v_1,\dots,v_M|x_1,\dots,x_N|a_{12})$
(see Figure \ref{ordinarypicturerecursion}).
If one sets $v_M=u_N$ and use
$
W \Bigg(
 \begin{array}{cc}
\overrightarrow{a} & \overrightarrow{a}+\hat{e}_1 \\
\overrightarrow{a}+\hat{e}_2 & \overrightarrow{a}+\hat{e}_1+\hat{e}_2
 \end{array} 
\Bigg|
u_N
\Bigg|
v_M
\Bigg)\Bigg|_{v_M=u_N}
=0
$,
one first finds that the state vectors around the face
of the top right corner freeze.
Continuing graphical considerations, one finds that
all the state vectors at the rightmost column and the top row freeze.
One can also see that the state of the top left corner
of remaining unfreezed part is now $\overrightarrow{a_{12}}+\hat{e}_1$,
thus the unfreezed part is the partition functions
$W_{M-1,N-1}(u_1,\dots,u_{N-1}|v_1,\dots,v_{M-1}|x_1,\dots,x_{N-1}|a_{12}+1)$.
Hence, the partition functions of the face model
$W_{M,N}(u_1,\dots,u_{N}|v_1,\dots,v_{M}|x_1,\dots,x_{N}|a_{12})$
evaluated at $v_M=u_N$ is given as the product of the weights
of the freezed faces
and a smaller partition functions
$W_{M-1,N-1}(u_1,\dots,u_{N-1}|v_1,\dots,v_{M-1}|x_1,\dots,x_{N-1}|a_{12}+1)$,
and we have the following relation
\begin{align}
&W_{M,N}(u_1,\dots,u_N|v_1,\dots,v_M|x_1,\dots,x_N|a_{12})
|_{v_M=u_N}
\nonumber \\
=&\prod_{j=1}^{N-1} \frac{[1-u_j+u_N]}{[1]}
\prod_{j=1}^{M-1} \frac{[1+u_N-v_j]}{[1]}
\nonumber \\
&\times W_{M-1,N-1}(u_1,\dots,u_{N-1}|v_1,\dots,v_{M-1}|x_1,\dots,x_{N-1}
|a_{12}+1)
.
\end{align}

Property (3) for the case $x_N \neq M$ can also be
easily shown from its graphical representation
(Figure \ref{ordinarypicturerecursion2}).
In this case, one finds the rightmost column freeze,
and the remaining unfreezed part is the partition functions
$W_{M-1,N}(u_1,\dots,u_{N}|v_1,\dots,v_{M-1}|x_1,\dots,x_{N}|a_{12})$.
Thus, we find that the partition functions
$W_{M,N}(u_1,\dots,u_{N}|v_1,\dots,v_{M}|x_1,\dots,x_{N}|a_{12})$
is the product of a smaller partition functions
$W_{M-1,N}(u_1,\dots,u_{N}|v_1,\dots,v_{M-1}|x_1,\dots,x_{N}|a_{12})$
and the product of the weights of the freezed faces
at the rightmost column
\begin{align}
&W \Bigg(
 \begin{array}{cc}
\overrightarrow{a_{12}}+(M-1)\hat{e}_1+(j-1)\hat{e}_2 & \overrightarrow{a_{12}}
+M\hat{e}_1+(j-1)\hat{e}_2 \\
\overrightarrow{a_{12}}+(M-1)\hat{e}_1+j\hat{e}_2 & \overrightarrow{a_{12}}
+M\hat{e}_1+j\hat{e}_2
 \end{array} 
\Bigg|
u_{N+1-j}
\Bigg|
v_M
\Bigg) \nonumber \\
=&\frac{[u_{N+1-j}-v_M][a_{21}-(M-1)-j]}{[1][a_{21}-(M-1)-j+1]},
\ \ \ j=1,\dots,N,
\end{align}
and we get the following factorization
\begin{align}
&W_{M,N}(u_1,\dots,u_N|v_1,\dots,v_M|x_1,\dots,x_N|a_{12})
 \nonumber \\
=&\prod_{j=1}^N \frac{[u_{N+1-j}-v_M][a_{21}-(M-1)-j]}{[1][a_{21}-(M-1)-j+1]}
W_{M-1,N}(u_1,\dots,u_N|v_1,\dots,v_{M-1}|x_1,\dots,x_N) \nonumber \\
=&\frac{[a_{21}-M-N+1]}{[a_{21}-M+1]}
\prod_{j=1}^N \frac{[u_j-v_M]}{[1]}
W_{M-1,N}(u_1,\dots,u_N|v_1,\dots,v_{M-1}|x_1,\dots,x_N|a_{12})
.
\end{align}
This shows Property (3) for the case $x_N \neq M$.

Property (4) can also be proved in the same way with Property (3).
One immediately sees that all the state vectors freeze, and multiplying 
all the weights of the freezed faces gives
\begin{align}
&
W_{M,1}(u|v_1,\dots,v_M|M|a_{12})=
\frac{[-u+v_M+a_{12}+M-1]}{[a_{12}+M-1]}
\prod_{k=1}^{M-1} \frac{[1+u-v_k]}{[1]}.
\end{align}
\end{proof}

\section{Partition functions and Elliptic multivariable functions}

In this section, we introduce the
elliptic Schur-type symmetric function.
The elliptic symmetric function,
together with the elliptic factors represents the partition functions
of the elliptic supersymmetric integrable model.

\begin{definition}
We define the following elliptic symmetric function \\
$E_{M,N}(u_1,\dots,u_N|v_1,\dots,v_M|x_1,\dots,x_N|a_{12})$
which depends on the symmetric variables \\
$u_1,\dots,u_N$,
complex parameters $v_1,\dots,v_M$, a complex parameter $a_{12}$
and integers $x_1,\dots,x_N$ satisfying
$1 \le x_1 < \cdots < x_N \le M$,
\begin{align}
&E_{M,N}(u_1,\dots,u_N|v_1,\dots,v_M|x_1,\dots,x_N|a_{12})
\nonumber \\
=&
\prod_{k=1}^N \prod_{x_{k+1}^M-2 \ge j \ge x_k^M}
\frac{[a_{12}+j+N]}{[a_{12}+j+N-k]}
\sum_{\sigma \in S_N}
\prod_{1 \le j < k \le N} \frac{[1]}{[u_{\sigma(j)}-u_{\sigma(k)}]}
\prod_{j=1}^N \prod_{k=x_j+1}^M \frac{[u_{\sigma(j)}-v_k]}{[1]}
\nonumber \\
&\times
\prod_{j=1}^N \frac{[-u_{\sigma(j)}+v_{x_j}+a_{12}+x_j+N-2]}{[a_{12}+x_j+N-2]}
\prod_{j=1}^N \prod_{k=1}^{x_j-1}
\frac{[1+u_{\sigma(j)}-v_k]}{[1]}.
\label{ordinaryrighthandside}
\end{align}
Here, $x_k^M$, $k=1,\dots,N+1$ is defined as $x_{N+1}^M=M+1$
and $x_k^M=x_k$, $k=1,\dots,N$.
\end{definition}
One notes that the expression of the multivariable elliptic function
\eqref{ordinaryrighthandside}
can also be written in the Schur-like determinant form
\begin{align}
&E_{M,N}(u_1,\dots,u_N|v_1,\dots,v_M|x_1,\dots,x_N|a_{12})
\nonumber \\
=&
\prod_{k=1}^N \prod_{x_{k+1}^M-2 \ge j \ge x_k^M}
\frac{[a_{12}+j+N]}{[a_{12}+j+N-k]}
\prod_{1 \le j < k \le N} \frac{[1]}{[u_{j}-u_{k}]}
\mathrm{det}_N (f_{x_j}(u_k|v_1,\dots,v_M)), \\
&f_{x}(u|v_1,\dots,v_M)=
\frac{[-u+v_x+a_{12}+x+N-2]}{[a_{12}+x+N-2]}
\prod_{k=1}^{x-1} \frac{[1+u-v_k]}{[1]}
\prod_{k=x+1}^M \frac{[u-v_k]}{[1]}.
\end{align}
We now state the main theorem of this paper.

\begin{theorem} \label{maintheoremstatement}
The partition functions of the elliptic supersymmetric integrable model
\\
$W_{M,N}(u_1,\dots,u_N|v_1,\dots,v_M|x_1,\dots,x_N|a_{12})$
is explicitly expressed as the
product of elliptic factors
$\displaystyle \prod_{1 \le j < k \le N} \frac{[1+u_k-u_j]}{[1]}$ and 
the elliptic symmetric function \\
$E_{M,N}(u_1,\dots,u_N|v_1,\dots,v_M|x_1,\dots,x_N|a_{12})$
\begin{align}
&W_{M,N}(u_1,\dots,u_N|v_1,\dots,v_M|x_1,\dots,x_N|a_{12}) \nonumber \\
=&
\prod_{1 \le j < k \le N} \frac{[1+u_k-u_j]}{[1]}
E_{M,N}(u_1,\dots,u_N|v_1,\dots,v_M|x_1,\dots,x_N|a_{12}).
\label{maintheorem}
\end{align}
\end{theorem}
We give the proof of Theorem \ref{maintheoremstatement}
in the next section.
The correspondence between the partition functions of integrable models
and the product of factors and symmetric functions was first
obtained for the trigonometric free-fermionic vertex models
by Bump-Brubaker-Friedberg \cite{BBF}
and Bump-McNamara-Nakasuji \cite{BMN}.
They showed that the partition functions are given as
the product of the deformed Vandermonde
determinant and (factorial) Schur functions,
from which the Tokuyama formula was realized in the language
of integrable vertex models.
Theorem \ref{maintheoremstatement} can be regarded as an elliptic analog
of the correspondence in some sense.

To explain the meaning of this, let us apply
the idea of Bump-Brubaker-Friedberg \cite{BBF}
and Bump-McNamara-Nakasuji \cite{BMN}.
The idea goes as follows.
First, we will show in the next section that
the partition functions $W_{M,N}(u_1,\dots,u_N|v_1,\dots,v_M|x_1,\dots,x_N|a_{12})$
are expressed using elliptic Vandermonde determinants
and elliptic Schur functions as \eqref{maintheorem}.

On the other hand, there is a direct way to evaluate the partition functions
from its definition (Definition \ref{definitionofwavefunctions}).
For example, when we directly compute the partition function \\
$W_{3,2}(u_1,u_2,u_3|v_1,v_2|1,3|a_{12})$ from its definition,
we find that there are exactly three inner states making non-zero contributions
to the partition function, and we have
\begin{align}
&W_{3,2}(u_1,u_2,u_3|v_1,v_2|1,3|a_{12}) \nonumber \\
=&\frac{[1+u_2-v_1]}{[1]} \frac{[1+u_2-v_2]}{[1]} \frac{[a_{12}+2-u_2+v_3]}{[a_{12}+2]} \nonumber \\
\times&\frac{[a_{12}+1-u_1+v_1]}{[a_{12}+1]} \frac{[u_1-v_2][-a_{12}-3]}{[1][-a_{12}-2]} \frac{[1-u_1+v_3]}{[1]} \nonumber \\
+&\frac{[1+u_2-v_1]}{[1]} \frac{[a_{12}+1-u_2+v_2]}{[a_{12}+1]} \frac{[u_2-v_3][-a_{12}-3]}{[1][-a_{12}-2]} \nonumber \\
\times&\frac{[a_{12}+1-u_1+v_1]}{[a_{12}+1]} \frac{[-a_{12}-2-u_1+v_2]}{[-a_{12}-2]} \frac{[a_{12}+3-u_1+v_3]}{[a_{12}+3]} \nonumber \\
+&\frac{[a_{12}-u_2+v_1]}{[a_{12}]} \frac{[u_2-v_2][-a_{12}-2]}{[1][-a_{12}-1]} \frac{[u_2-v_3][-a_{12}-3]}{[1][-a_{12}-2]} \nonumber \\
\times&\frac{[u_1-v_1][a_{12}]}{[1][a_{12}+1]} \frac{[1+u_1-v_2]}{[1]} \frac{[a_{12}+3-u_1+v_3]}{[a_{12}+3]}. \label{useforcomparison}
\end{align}
Combining
\eqref{maintheorem}
and
\eqref{useforcomparison},
we get an expansion expression for the
product of elliptic factors
$\displaystyle \prod_{1 \le j < k \le 3} \frac{[1+u_k-u_j]}{[1]}$ and 
the elliptic symmetric function
$E_{3,2}(u_1,u_2,u_3|v_1,v_2|1,3|a_{12})$
\begin{align}
&\prod_{1 \le j < k \le 3} \frac{[1+u_k-u_j]}{[1]}
E_{3,2}(u_1,u_2,u_3|v_1,v_2|1,3|a_{12}) \nonumber \\
=&\frac{[1+u_2-v_1]}{[1]} \frac{[1+u_2-v_2]}{[1]} \frac{[a_{12}+2-u_2+v_3]}{[a_{12}+2]} \nonumber \\
\times&\frac{[a_{12}+1-u_1+v_1]}{[a_{12}+1]} \frac{[u_1-v_2][-a_{12}-3]}{[1][-a_{12}-2]} \frac{[1-u_1+v_3]}{[1]} \nonumber \\
+&\frac{[1+u_2-v_1]}{[1]} \frac{[a_{12}+1-u_2+v_2]}{[a_{12}+1]} \frac{[u_2-v_3][-a_{12}-3]}{[1][-a_{12}-2]} \nonumber \\
\times&\frac{[a_{12}+1-u_1+v_1]}{[a_{12}+1]} \frac{[-a_{12}-2-u_1+v_2]}{[-a_{12}-2]} \frac{[a_{12}+3-u_1+v_3]}{[a_{12}+3]} \nonumber \\
+&\frac{[a_{12}-u_2+v_1]}{[a_{12}]} \frac{[u_2-v_2][-a_{12}-2]}{[1][-a_{12}-1]} \frac{[u_2-v_3][-a_{12}-3]}{[1][-a_{12}-2]} \nonumber \\
\times&\frac{[u_1-v_1][a_{12}]}{[1][a_{12}+1]} \frac{[1+u_1-v_2]}{[1]} \frac{[a_{12}+3-u_1+v_3]}{[a_{12}+3]}.
\end{align}
We remark that the original Tokuyama formula
is an expansion formula for the product of a
one-parameter deformation of the factorization of
the Vandermonde determinant, whereas
a free deformation parameter is absent in the elliptic Vandermonde determinant
$\displaystyle \prod_{1 \le j < k \le N} \frac{[1+u_k-u_j]}{[1]}$
in the partition functions of the elliptic Deguchi-Martin model.

\section{Proof of Theorem \ref{maintheoremstatement}}
We give the proof of
Theorem \ref{maintheoremstatement}
in this section. \\
Let us denote the right hand side of \eqref{maintheorem} as
$G_{M,N}(u_1,\dots,u_N|v_1,\dots,v_M|x_1,\dots,x_N|a_{12})$
\begin{align}
&G_{M,N}(u_1,\dots,u_N|v_1,\dots,v_M|x_1,\dots,x_N|a_{12}) \nonumber \\
:=&
\prod_{1 \le j < k \le N} \frac{[1+u_k-u_j]}{[1]}
E_{M,N}(u_1,\dots,u_N|v_1,\dots,v_M|x_1,\dots,x_N|a_{12})
\nonumber \\
=&\prod_{1 \le j < k \le N} \frac{[1+u_k-u_j]}{[1]}
\prod_{k=1}^N \prod_{x_{k+1}^M-2 \ge j \ge x_k^M}
\frac{[a_{12}+j+N]}{[a_{12}+j+N-k]} \nonumber \\
&\times \sum_{\sigma \in S_N}
\prod_{1 \le j < k \le N} \frac{[1]}{[u_{\sigma(j)}-u_{\sigma(k)}]}
\prod_{j=1}^N \prod_{k=x_j+1}^M \frac{[u_{\sigma(j)}-v_k]}{[1]}
\nonumber \\
&\times
\prod_{j=1}^N \frac{[-u_{\sigma(j)}+v_{x_j}+a_{12}+x_j+N-2]}{[a_{12}+x_j+N-2]}
\prod_{j=1}^N \prod_{k=1}^{x_j-1}
\frac{[1+u_{\sigma(j)}-v_k]}{[1]}.
\end{align}
We prove Theorem \ref{maintheoremstatement} by showing that
$G_{M,N}(u_1,\dots,u_N|v_1,\dots,v_M|x_1,\dots,x_N|a_{12})$
satisfies all the four Properties in
Proposition \ref{ordinarypropertiesfordomainwallboundarypartitionfunction}
obtained from the Izergin-Korepin analysis.

\begin{lemma} \label{lemmaone}
$G_{M,N}(u_1,\dots,u_N|v_1,\dots,v_M|x_1,\dots,x_N|a_{12})$
satisfies Property (1) in Proposition
\ref{ordinarypropertiesfordomainwallboundarypartitionfunction}.
\end{lemma}
\begin{proof}
We first collect all the factors which depend on $v_M$
in \\
$G_{M,N}(u_1,\dots,u_N|v_1,\dots,v_M|x_1,\dots,x_N|a_{12})$ when $x_N=M$.
For each summand, we find the factors are
$\displaystyle \prod_{j=1}^{N-1} [u_{\sigma(j)}-v_M]$ and
$\displaystyle [-u_{\sigma(N)}+v_M+a_{12}+N+M-2]$.
Let us denote the product of the factors as $f_\sigma(v_M)$:
\begin{align}
f_\sigma(v_M)=[-u_{\sigma(N)}+v_M+a_{12}+N+M-2]
\prod_{j=1}^{N-1} [u_{\sigma(j)}-v_M].
\end{align}
Using \eqref{qpuseone} and \eqref{qpusetwo},
we can calculate the quasi-periodicities of
the function $f_\sigma(v_M)$
\begin{align}
f_\sigma(v_M+2K_1/\lambda)
&=(-1)^N f_\sigma(v_M), \label{qpsummandone} \\
f_\sigma(v_M+2 i K_2/\lambda)&
=(-q^{-1})^N
\mathrm{exp} \Bigg(-\frac{i \pi\lambda}{K_1} 
\Bigg(Nv_M-\sum_{j=1}^N u_j+a_{12}+N+M-2 \Bigg)
\Bigg) f_\sigma(v_M). \label{qpsummandtwo}
\end{align}

The quasi-periodicities \eqref{qpsummandone} and \eqref{qpsummandtwo}
do not depend on the permutation $\sigma$,
hence all the summands of
$G_{M,N}(u_1,\dots,u_N|v_1,\dots,v_M|x_1,\dots,x_N|a_{12})$
have the same quasi-periodicities, and we get
the quasi-periodicities for
$G_{M,N}(u_1,\dots,u_N|v_1,\dots,v_M|x_1,\dots,x_N|a_{12})$
\begin{align}
&G_{M,N}(u_1,\dots,u_N|v_1,\dots,v_M+2K_1/\lambda|x_1,\dots,x_N|a_{12})
\nonumber \\
=&(-1)^N G_{M,N}(u_1,\dots,u_N|v_1,\dots,v_M|x_1,\dots,x_N|a_{12}),
\label{symmetricfunctionsqpfirst} \\
&
G_{M,N}(u_1,\dots,u_N|v_1,\dots,v_M+2 i K_2/\lambda|x_1,\dots,x_N|a_{12})
\nonumber \\
=&(-q^{-1})^N
\mathrm{exp} \Bigg(-\frac{i \pi\lambda}{K_1} 
\Bigg(Nv_M-\sum_{j=1}^N u_j+a_{12}+N+M-2 \Bigg)
\Bigg) \nonumber \\
&\times G_{M,N}(u_1,\dots,u_N|v_1,\dots,v_M|x_1,\dots,x_N|a_{12})
. \label{symmetricfunctionsqpsecond}
\end{align}
The quasi-periodicities \eqref{symmetricfunctionsqpfirst}
and \eqref{symmetricfunctionsqpsecond}
for
$G_{M,N}(u_1,\dots,u_N|v_1,\dots,v_M|x_1,\dots,x_N|a_{12})$
are exactly the same with the quasi-periodicities
\eqref{wavefunctionqpfirst} and \eqref{wavefunctionqpsecond}
for the partition functions \\
$W_{M,N}(u_1,\dots,u_N|v_1,\dots,v_M|x_1,\dots,x_N|a_{12})$,
hence Property (1) is proved.

\end{proof}

\begin{lemma} \label{lemmatwo}
$G_{M,N}(u_1,\dots,u_N|v_1,\dots,v_M|x_1,\dots,x_N|a_{12})$
satisfies Property (2) in Proposition
\ref{ordinarypropertiesfordomainwallboundarypartitionfunction}.
\end{lemma}

\begin{proof}
One first notes by its definition
\eqref{ordinaryrighthandside} that
$E_{M,N}(u_1,\dots,u_N|v_1,\dots,v_M|x_1,\dots,x_N|a_{12})$,
which consists
a part of $G_{M,N}(u_1,\dots,u_N|v_1,\dots,v_M|x_1,\dots,x_N|a_{12})$,
is a symmetric function of $u_1,\dots,u_N$ which means
\begin{align}
&E_{M,N}(u_1,\dots,u_N|v_1,\dots,v_M|x_1,\dots,x_N|a_{12})
\nonumber \\
=&E_{M,N}(u_{\sigma(1)},\dots,u_{\sigma(N)}|v_1,\dots,v_M|x_1,\dots,x_N|a_{12}).
\label{forpermutationone}
\end{align}
Next, looking at the other factor 
$\displaystyle \prod_{1 \le j < k \le N} \frac{[1+u_k-u_j]}{[1]}$
constructing the function \\
$G_{M,N}(u_1,\dots,u_N|v_1,\dots,v_M|x_1,\dots,x_N|a_{12})$,
we easily find the following equality
\begin{align}
&\prod_{\substack{1 \le j < k \le N \\ \sigma(j) > \sigma(k)}}
[1+u_{\sigma(k)}-u_{\sigma(j)}]
\prod_{1 \le j < k \le N}
\frac{[1+u_k-u_j]}{[1]}
\nonumber \\
=&
\prod_{\substack{1 \le j < k \le N \\ \sigma(j) > \sigma(k)}}
[1+u_{\sigma(j)}-u_{\sigma(k)}]
\prod_{1 \le j < k \le N}
\frac{[1+u_{\sigma(k)}-u_{\sigma(j)}]}{[1]}.
\label{forpermutationtwo}
\end{align}
Since
$G_{M,N}(u_1,\dots,u_N|v_1,\dots,v_M|x_1,\dots,x_N|a_{12})$
is defined as the product of \\
$\displaystyle \prod_{1 \le j < k \le N}\frac{[1+u_k-u_j]}{[1]}$
and $E_{M,N}(u_1,\dots,u_N|v_1,\dots,v_M|x_1,\dots,x_N|a_{12})$,
one has the following relation as a combination of
\eqref{forpermutationone} and \eqref{forpermutationtwo}
\begin{align}
&\prod_{\substack{1 \le j < k \le N \\ \sigma(j) > \sigma(k)}}
[1+u_{\sigma(k)}-u_{\sigma(j)}]
G_{M,N}(u_1,\dots,u_N|v_1,\dots,v_M|x_1,\dots,x_N|a_{12}) \nonumber \\
=&
\prod_{\substack{1 \le j < k \le N \\ \sigma(j) > \sigma(k)}}
[1+u_{\sigma(j)}-u_{\sigma(k)}]
G_{M,N}(u_{\sigma(1)},\dots,u_{\sigma(N)}|v_1,\dots,v_M|x_1,\dots,x_N|a_{12}).
\end{align}
Thus, the multivariable functions
$G_{M,N}(u_1,\dots,u_N|v_1,\dots,v_M|x_1,\dots,x_N|a_{12})$
and the partition functions
$W_{M,N}(u_1,\dots,u_N|v_1,\dots,v_M|x_1,\dots,x_N|a_{12})$
satisfy the same exchange relation, hence Property (2) is proved.

\end{proof}

\begin{lemma} \label{lemmathree}
$G_{M,N}(u_1,\dots,u_N|v_1,\dots,v_M|x_1,\dots,x_N|a_{12})$
satisfies Property (3) in Proposition
\ref{ordinarypropertiesfordomainwallboundarypartitionfunction}.
\end{lemma}
\begin{proof}
We first treat the case $x_N=M$.
We show that the function \\
$G_{M,N}(u_1,\dots,u_N|v_1,\dots,v_M|x_1,\dots,x_N|a_{12})$ satisfies
\eqref{ordinaryrecursionwavefunction}.
First, one notes that the expression of the factor
\begin{align}
\prod_{j=1}^N \prod_{k=x_j+1}^M
\frac{[u_{\sigma(j)}-v_k]}{[1]},
\end{align}
in each summand is redundant when $x_N=M$. It is
essentially
\begin{align}
\prod_{j=1}^{N-1} \prod_{k=x_j+1}^M
\frac{[u_{\sigma(j)}-v_k]}{[1]}.
\label{ordinaryfactorconsideration}
\end{align}
Looking at the part
$\displaystyle
\prod_{j=1}^{N-1} \frac{[u_{\sigma(j)}-v_M]}{[1]}
$ in \eqref{ordinaryfactorconsideration},
one finds this factor vanishes unless $\sigma$ satisfies $\sigma(N)=N$
if one substitutes $v_M=u_N$.

Therefore, only the summands satisfying $\sigma(N)=N$
in \eqref{ordinaryrighthandside} survive
after the substitution $v_M=u_N$.
Note also that the factor
\begin{align}
\prod_{k=1}^N \prod_{x_{k+1}^M-2 \ge j \ge x_k^M}
\frac{[a_{12}+j+N]}{[a_{12}+j+N-k]},
\end{align}
is essentially
\begin{align}
\prod_{k=1}^{N-1} \prod_{x_{k+1}^M-2 \ge j \ge x_k^M}
\frac{[a_{12}+j+N]}{[a_{12}+j+N-k]},
\end{align}
for the case $x_N=M$ since $x_{N+1}^{M}=M+1$
and $x_N^{M}=M$.

From the above considerations, we find that we
can rewrite the multivariable function
$G_{M,N}(u_1,\dots,u_N|v_1,\dots,v_M|x_1,\dots,x_N|a_{12})$
evaluated at $v_M=u_N$
by using the symmetric group $S_{N-1}$
where every $\sigma^\prime \in S_{N-1}$ satisfies
$\{\sigma^\prime(1),\cdots,\sigma^\prime(N-1)\}=\{1,\cdots,N-1 \}$ as follows:
\begin{align}
&G_{M,N}(u_1,\dots,u_N|v_1,\dots,v_M|x_1,\dots,x_N|a_{12})
|_{v_M=u_N} \nonumber \\
=&\prod_{1 \le j < k \le N-1} \frac{[1+u_k-u_j]}{[1]}
\prod_{j=1}^{N-1} \frac{[1+u_N-u_j]}{[1]} 
\prod_{k=1}^{N-1} \prod_{x_{k+1}^M-2 \ge j \ge x_k^M}
\frac{[a_{12}+j+N]}{[a_{12}+j+N-k]}
\nonumber \\
&\times \sum_{\sigma^\prime \in S_{N-1}}
\prod_{1 \le j < k \le N-1}
\frac{[1]}{[u_{\sigma^\prime(j)}-u_{\sigma^\prime(k)}]}
\prod_{j=1}^{N-1} \frac{[1]}{[u_{\sigma^\prime(j)}-u_N]} \nonumber \\
&\times
\prod_{j=1}^{N-1} \prod_{k=x_j+1}^{M-1}
\frac{[u_{\sigma^\prime(j)}-v_k]}{[1]}
\prod_{j=1}^{N-1}
\frac{[u_{\sigma^\prime(j)}-u_N]}{[1]}
\nonumber \\
&\times
\prod_{j=1}^{N-1}
\frac{[-u_{\sigma^\prime(j)}+v_{x_j}+(a_{12}+1)+x_j+(N-1)-2]}
{[(a_{12}+1)+x_j+(N-1)-2]}
\nonumber \\
&\times
\prod_{j=1}^{N-1} \prod_{k=1}^{x_j-1}
\frac{[1+u_{\sigma^\prime(j)}-v_k]}{[1]}
\prod_{k=1}^{M-1}
\frac{[1+u_N-v_k]}{[1]}.
\label{ordinaryrighthandsideaftersubstitution}
\end{align}
Rewriting a factor in \eqref{ordinaryrighthandsideaftersubstitution} as
\begin{align}
\prod_{k=1}^{N-1} \prod_{x_{k+1}^M-2 \ge j \ge x_k^M}
\frac{[a_{12}+j+N]}{[a_{12}+j+N-k]}
=
\prod_{k=1}^{N-1} \prod_{x_{k+1}^{M-1}-2 \ge j \ge x_k^{M-1}}
\frac{[(a_{12}+1)+j+(N-1)]}{[(a_{12}+1)+j+(N-1)-k]},
\end{align}
and cancelling same factors appearing in the denominator and numerator,
one finds that \eqref{ordinaryrighthandsideaftersubstitution}
can be further simplified as
\begin{align}
&G_{M,N}(u_1,\dots,u_N|v_1,\dots,v_M|x_1,\dots,x_N|a_{12})
|_{v_M=u_N} \nonumber \\
=&
\prod_{j=1}^{N-1} \frac{[1+u_N-u_j]}{[1]}
\prod_{k=1}^{M-1}
\frac{[1+u_N-v_k]}{[1]} \nonumber \\
&\times \prod_{1 \le j < k \le N-1} \frac{[1+u_k-u_j]}{[1]}
\prod_{k=1}^{N-1} \prod_{x_{k+1}^{M-1}-2 \ge j \ge x_k^{M-1}}
\frac{[(a_{12}+1)+j+(N-1)]}{[(a_{12}+1)+j+(N-1)-k]}
\nonumber \\
&\times \sum_{\sigma^\prime \in S_{N-1}}
\prod_{1 \le j < k \le N-1}
\frac{[1]}{[u_{\sigma^\prime(j)}-u_{\sigma^\prime(k)}]}
\prod_{j=1}^{N-1} \prod_{k=x_j+1}^{M-1}
\frac{[u_{\sigma^\prime(j)}-v_k]}{[1]}
\nonumber \\
&\times
\prod_{j=1}^{N-1}
\frac{[-u_{\sigma^\prime(j)}+v_{x_j}+(a_{12}+1)+x_j+(N-1)-2]}
{[(a_{12}+1)+x_j+(N-1)-2]}
\prod_{j=1}^{N-1} \prod_{k=1}^{x_j-1}
\frac{[1+u_{\sigma^\prime(j)}-v_k]}{[1]}.
\label{ordinaryrighthandsideaftersubstitutiontwo}
\end{align}
Finally, noting
\begin{align}
& \prod_{1 \le j < k \le N-1} \frac{[1+u_k-u_j]}{[1]}
\prod_{k=1}^{N-1} \prod_{x_{k+1}^{M-1}-2 \ge j \ge x_k^{M-1}}
\frac{[(a_{12}+1)+j+(N-1)]}{[(a_{12}+1)+j+(N-1)-k]}
\nonumber \\
&\times \sum_{\sigma^\prime \in S_{N-1}}
\prod_{1 \le j < k \le N-1}
\frac{[1]}{[u_{\sigma^\prime(j)}-u_{\sigma^\prime(k)}]}
\prod_{j=1}^{N-1} \prod_{k=x_j+1}^{M-1}
\frac{[u_{\sigma^\prime(j)}-v_k]}{[1]}
\nonumber \\
&\times
\prod_{j=1}^{N-1}
\frac{[-u_{\sigma^\prime(j)}+v_{x_j}+(a_{12}+1)+x_j+(N-1)-2]}
{[(a_{12}+1)+x_j+(N-1)-2]}
\prod_{j=1}^{N-1} \prod_{k=1}^{x_j-1}
\frac{[1+u_{\sigma^\prime(j)}-v_k]}{[1]}
\nonumber \\
=&
G_{M-1,N-1}(u_1,\dots,u_{N-1}|v_1,\dots,v_{M-1}|x_1,\dots,x_{N-1}|a_{12}+1)
,
\end{align}
one finds that
\eqref{ordinaryrighthandsideaftersubstitutiontwo} is exactly the
following recursion relation
for the function \\
$G_{M,N}(u_1,\dots,u_N|v_1,\dots,v_M|x_1,\dots,x_N|a_{12})$
\begin{align}
&G_{M,N}(u_1,\dots,u_N|v_1,\dots,v_M|x_1,\dots,x_N|a_{12})
|_{v_M=u_N} \nonumber \\
=&
\prod_{j=1}^{N-1} \frac{[1+u_N-u_j]}{[1]}
\prod_{k=1}^{M-1}
\frac{[1+u_N-v_k]}{[1]} \nonumber \\
&\times
G_{M-1,N-1}(u_1,\dots,u_{N-1}|v_1,\dots,v_{M-1}|x_1,\dots,x_{N-1}|a_{12}+1)
.
\end{align}
Thus we have shown that the functions
$G_{M,N}(u_1,\dots,u_N|v_1,\dots,v_M|x_1,\dots,x_N|a_{12})$
satisfy the same relation \eqref{ordinaryrecursionwavefunction}
with the partition functions
$W_{M,N}(u_1,\dots,u_N|v_1,\dots,v_M|x_1,\dots,x_N|a_{12})$,
hence Property (3) for the case $x_N=M$ is shown.

Next, we examine the case $x_N \neq M$.
This can be shown in a much simpler way.
We rewrite $G_{M,N}(u_1,\dots,u_N|v_1,\dots,v_M|x_1,\dots,x_N|a_{12})$ as
\begin{align}
&G_{M,N}(u_1,\dots,u_N|v_1,\dots,v_M|x_1,\dots,x_N|a_{12}) \nonumber \\
=&\prod_{1 \le j < k \le N} \frac{[1+u_k-u_j]}{[1]}
\frac{[a_{12}+M-1+N]}{[a_{12}+M-1]}
\prod_{k=1}^{N} \prod_{x_{k+1}^{M-1}-2 \ge j \ge x_k^{M-1}}
\frac{[a_{12}+j+N]}{[a_{12}+j+N-k]}
\nonumber \\
&\times \sum_{\sigma \in S_N}
\prod_{1 \le j < k \le N} \frac{[1]}{[u_{\sigma(j)}-u_{\sigma(k)}]}
\prod_{j=1}^N \prod_{k=x_j+1}^{M-1}
\frac{[u_{\sigma(j)}-v_k]}{[1]}
\prod_{j=1}^N \frac{[u_{\sigma(j)}-v_M]}{[1]}
\nonumber \\
&\times
\prod_{j=1}^N \frac{[-u_{\sigma(j)}+v_{x_j}+a_{12}+x_j-N-2]}{[a_{12}+x_j+N-2]}
\prod_{j=1}^N \prod_{k=1}^{x_j-1}
\frac{[1+u_{\sigma(j)}-v_k]}{[1]}.
\label{originalrighthandsidenew}
\end{align}
We next use the obvious identity
\begin{align}
\displaystyle 
\prod_{j=1}^N \frac{[u_{\sigma(j)}-v_M]}{[1]}
=\prod_{j=1}^N \frac{[u_{j}-v_M]}{[1]},
\end{align}
to get this factor out of the sum in
\eqref{originalrighthandsidenew}, and we find
\begin{align}
&G_{M,N}(u_1,\dots,u_N|v_1,\dots,v_M|x_1,\dots,x_N|a_{12}) \nonumber \\
=&\frac{[a_{21}-M-N+1]}{[a_{21}-M+1]}
\prod_{j=1}^N \frac{[u_{j}-v_M]}{[1]}
\nonumber \\
&\times \prod_{1 \le j < k \le N} \frac{[1+u_k-u_j]}{[1]} 
\prod_{k=1}^{N} \prod_{x_{k+1}^{M-1}-2 \ge j \ge x_k^{M-1}}
\frac{[a_{12}+j+N]}{[a_{12}+j+N-k]}
\nonumber \\
&\times \sum_{\sigma \in S_N}
\prod_{1 \le j < k \le N} \frac{[1]}{[u_{\sigma(j)}-u_{\sigma(k)}]}
\prod_{j=1}^N \prod_{k=x_j+1}^{M-1}
\frac{[u_{\sigma(j)}-v_k]}{[1]}
\nonumber \\
&\times
\prod_{j=1}^N \frac{[-u_{\sigma(j)}+v_{x_j}+a_{12}+x_j-N-2]}{[a_{12}+x_j+N-2]}
\prod_{j=1}^N \prod_{k=1}^{x_j-1}
\frac{[1+u_{\sigma(j)}-v_k]}{[1]} \nonumber \\
=&\frac{[a_{21}-M-N+1]}{[a_{21}-M+1]}
\prod_{j=1}^N \frac{[u_{j}-v_M]}{[1]}
G_{M-1,N}(u_1,\dots,u_N|v_1,\dots,v_{M-1}|x_1,\dots,x_N|a_{12}).
\label{samefactorization}
\end{align}
Note here that we have also used $a_{12}=-a_{21}$ and $[-u]=-[u]$.
\eqref{samefactorization} is the exactly the same recursion relation
for the partition functions
$W_{M,N}(u_1,\dots,u_N|v_1,\dots,v_M|x_1,\dots,x_N|a_{12})$
must satisfy for the case $x_N \neq M$.

\end{proof}

\begin{lemma} \label{lemmafour}
$G_{M,N}(u_1,\dots,u_N|v_1,\dots,v_M|x_1,\dots,x_N|a_{12})$
satisfies Property (4) in Proposition
\ref{ordinarypropertiesfordomainwallboundarypartitionfunction}.
\end{lemma}
\begin{proof}
It is trivial to check from the definition
of $G_{M,N}(u_1,\dots,u_N|v_1,\dots,v_M|x_1,\dots,x_N|a_{12})$
\eqref{ordinaryrighthandside}.
\end{proof}
{\it Proof of Theorem \ref{maintheoremstatement}.}
From Lemmas \ref{lemmaone}, \ref{lemmatwo}, \ref{lemmathree}
and \ref{lemmafour}, we find the elliptic function
$G_{M,N}(u_1,\dots,u_N|v_1,\dots,v_M|x_1,\dots,x_N|a_{12})$
satisfies all the
Properties (1), (2), (3) and (4) in
Proposition \ref{ordinarypropertiesfordomainwallboundarypartitionfunction},
which means that it is nothing but the explicit form of the partition functions \\
$
W_{M,N}(u_1,\dots,u_N|v_1,\dots,v_M|x_1,\dots,x_N|a_{12})
=
G_{M,N}(u_1,\dots,u_N|v_1,\dots,v_M|x_1,\dots,x_N|a_{12})
$.
\\
\hfill $\Box$

\section{Conclusion}
In this paper, we introduced and investigated the
partition functions of the
Deguchi-Martin model by using the Izergin-Korepin analysis.
We viewed the partition functions as an elliptic polynomial
and determined the properties the partition functions.
We next proved that the partition functions are expressed
as a product of
elliptic factors and elliptic Schur-type symmetric functions.
This result resembles the ones for the trigonometric model
whose partition functions are proved by Bump-Brubaker-Friedberg in \cite{BBF}
and by Bump-McNamara-Nakasuji in \cite{BMN}
to be given as the product of a one-parameter deformation
of the Vandermonde determinant and the Schur functions
and factorial Schur functions, respectively.
The result obtained in this paper
can be viewed as an elliptic analogue of their results.
The result obtained in this paper can be used to
construct new algebraic identities for elliptic multivariable functions
which will be reported elsewhere.
To explore the connections with number theory
seems to be an interesting topic.

It is also interesting to extend the study of partition functions
of the elliptic Deguchi-Martin model to other boundary conditions.
For example, if one changes the boundary condition of
the free-fermionic model to the reflecting boundary condition,
other types of symmetric functions
such as the symplectic Schur functions
\cite{Iv,BBCG,dualsymplectic} appear.
Therefore, one can expect that elliptic analogues
of the symplectic Schur functions appear by generalizing
the models from trigonometric to elliptic ones.

\section*{Acknowledgements}
The author thanks Prof. Korepin for discussions about
the Izergin-Korepin method and warm encouragement.
The author also expresses his sincere gratitude to the referee
for numerous valuable suggestions to improve the manuscript.
This work was partially supported by grant-in-Aid
for Scientific Research (C) No. 16K05468.

\end{document}